\documentclass[journal, 12pt, onecolumn, draftclsnofoot]{IEEEtran}

\usepackage{float}
\usepackage{cite}
\usepackage{algorithmic}
\usepackage{graphicx}
\usepackage{textcomp}
\usepackage{amsmath,amsfonts, amssymb, graphicx,amsthm,epsfig}
\usepackage{bm,bbm} 
\usepackage{dsfont}
\usepackage{mathtools}
\usepackage{array,multirow}
\usepackage{tikz-network}
\usepackage{subfig}
\usepackage{float}

\def\BibTeX{{\rm B\kern-.05em{\sc i\kern-.025em b}\kern-.08em
    T\kern-.1667em\lower.7ex\hbox{E}\kern-.125emX}}

\newcommand{\fracpartial}[2]{\frac{\partial #1}{\partial  #2}}
\newcommand{\bxi}{\bm{\xi}}
\newcommand{\bpsi}{\bm{\psi}}
\newcommand{\bmu}{\bm{\mu}}
\newcommand{\bS}{\bm{\lambda}}

\newcommand{\cN}{\mathcal{N}}

\newcommand{\brho}{\bm{\rho}}

\newcommand{\bL}{\bm{x}}

\newcommand{\cV}{\mathcal{V}}
\newcommand{\bnu}{\bm{\nu}}

\newcommand{\cA}{\mathcal{A}}
\newcommand{\cG}{\mathcal{G}}
\newcommand{\cH}{\mathcal{H}}

\newcommand{\bzet}{\bm{z}}
\newcommand{\cF}{\mathcal{F}}

\newcommand{\bY}{\bm{Y}}
\newcommand{\br}{\bm{r}}

\newcommand{\bR}{\bm{R}}

\newcommand{\bu}{\bm{u}}
\newcommand{\bss}{\bm{s}}

\newcommand{\bx}{\bm{x}}
\newcommand \bX {\mathbf{X}}

\newcommand{\E}{\mathbb{E}}
\newcommand{\bQ}{\mathbb{Q}}

\renewcommand{\AA}{\mathrm{AA}}
\newcommand{\GA}{\mathrm{GA}}
\DeclareMathOperator{\T}{\mathsf{T}}
\DeclareMathOperator{\dob}{\mathsf{Dob}}

\newcommand{\bP}[1]{\mathbb{P}\left(#1\right)}

\renewcommand{\qedsymbol}{$\blacksquare$}
\newcommand{\putqed}{\hfill\qedsymbol}

\newtheorem{theorem}{Theorem}
\newtheorem{definition}{Definition}
\newtheorem{lemma}{Lemma}
\newtheorem{corollary}{Corollary}
\newtheorem{assumption}{Assumption}
\newtheorem{proposition}{Proposition}

\begin{document}
\title{On the Arithmetic and Geometric Fusion of Beliefs for Distributed Inference}

\author{Mert Kayaalp, Yunus \.Inan, Emre Telatar, Ali H. Sayed \vspace{0.5em}\\
\textit{{\normalsize \'{E}cole Polytechnique F\'{e}d\'{e}rale de Lausanne (EPFL)}}
\thanks{Corresponding author: M. Kayaalp, email: mert.kayaalp@epfl.ch. This work was supported in part by grant 205121-184999 from the Swiss National Science Foundation (SNSF).}}

\maketitle

\vspace{-3em}

\begin{abstract}
We study the asymptotic learning rates of belief vectors in a distributed hypothesis testing problem under linear and log-linear combination rules. We show that under both combination strategies, agents are able to learn the truth exponentially fast, with a faster rate under log-linear fusion. We examine the gap between the rates in terms of network connectivity and information diversity. We also provide closed-form expressions for special cases involving federated architectures and exchangeable networks.
\end{abstract}


\section{\textbf{Introduction}}
\label{sec:introduction}
Statistical inference is the problem of learning a hidden variable from partially informative observations. It pertains to estimation and decision-making tasks, which are essential for engineering design problems. Many modern designs are large-scale and hence inference is apportioned to smaller devices. These devices --- henceforward referred to as agents --- have data acquiring and processing capabilities, and are possibly spatially dispersed. Distributed inference deals with the collaborative working of these agents to infer an unknown phenomenon of interest. 

There are two main settings depending on the topology of agents. The first is a federated setting where agents send their information, e.g., raw data, opinion or belief about the hidden state, to a \emph{fusion center}. The fusion center then aggregates the received messages and performs deduction \cite{tsitsiklis1993}. For example, weather stations at different locations can measure barometric pressure or cloud coverage and send a statistic of the measurements to a centralized unit, e.g., National Weather Service. By combining and processing the information received, the service is able to provide weather forecasts. However, gathering all information at a central location is prone to failure and raises privacy issues. Peer-to-peer networks that rely only on localized interactions offer a remedy. These networks are \emph{fully} decentralized, i.e., there is no central unit. Agents perform local computations and combine their immediate neighbors' information \cite{Sayed14,sayed14proc}. Wireless sensor networks assigned to perform area/industrial monitoring or environmental sensing are examples of fully distributed multi-agent networks \cite{BARBAROSSA2014}.

A fundamental question that arises is how to aggregate information received from the peripheral agents, in federated architectures; or from the neighbors, in fully distributed architectures. There exists a variety of combination strategies and the optimal \emph{pooling} strategy depends heavily on the application at hand \cite{koliander2022fusion}. We focus on two widely used strategies: arithmetic averaging (AA) and geometric averaging (GA). While the former is shown to be better with fusion of random variables or point estimates (\(\nu\)-fusion), the latter is more common in applications involving fusion of probability mass/density functions (\(f\)-fusion) \cite{li2019second}. Combining probability distributions is superior in the sense that a full description about the unknown variable of interest is utilized rather than some statistics of it, such as the mean or the median. Furthermore, this procedure is also compatible with heterogeneous data models across the agents. Therefore, in this work, we evaluate and compare the performance of AA and GA for the fusion of \emph{beliefs}, i.e., distributions over the unknown quantity of interest. 

\subsection{Related Work}

\subsubsection{Distributed Inference}
Distributed inference can be broadly classified into three approaches based on the characteristics of the unknown variable: estimation for continuous variables \cite{Sayed14,sayed14proc,bordley1982,lopes2008,dedecius2017,luengo2018}, detection (hypothesis testing) for categorical variables \cite{Ahlswede1986,Han_1987,chen2006,Hamad2021,inan_fundamental_2021,bajovic2011,matta2016,acemoglu_2011,jadbabaie_2012,zhao_2012,jadbabaie2013information,nedic_2017,inan2022social,toghani2021communication,lalitha_2018,bordignon2021adaptive,shahrampour_2013}, and filtering for dynamic variables \cite{khan_2008,cattivelli2010,hu2012,hlinka_2012,bandyopadhyay_2018,kayaalp_2021}. AA and GA are commonly used and the distinction between them is present in all approaches. We focus on distributed hypothesis testing in this work. Nevertheless, our results can shed light on optimal pooling strategies for distributed estimation and filtering as well.

Hypothesis testing with a fusion center is generally studied under channel imperfections or rate constraints between agents and the central node \cite{tsitsiklis1993, Ahlswede1986,Han_1987,chen2006,Hamad2021,inan_fundamental_2021}. It is well-known that under independent and identically distributed (i.i.d.) data with perfect links and fully-connected topologies, the optimal strategy is to combine log-likelihood ratios of agents \cite{chernoff52}. Hence, references \cite{tsitsiklis1993,chen2006,inan_fundamental_2021} study transmitting log-likelihood ratios after suitable processing. Similarly, appropriate functions of log-likelihood ratios can be shared between agents over decentralized networks \cite{bajovic2011,matta2016}. \\

\subsubsection{Social Learning}
Given the difficulty in performing full-blown Bayesian inference over graphs, social learning is a non-Bayesian inference paradigm that is able to perform distributed hypothesis testing over networks with (possibly bounded) rational agents by relaxing the full Bayesian requirement while still leading to truth learning with probability 1 (e.g., \cite{lalitha_2018,matta2016,jadbabaie_2012,nedic_2017,bordignon2021adaptive}). Under social learning, \emph{beliefs}, i.e., local posterior distributions, over the set of hypotheses are exchanged in lieu of log-likelihood ratios. This practice helps when agents have heterogeneous data models. Furthermore, it helps modeling opinion formation over social networks. Fully Bayesian strategies to find global posterior distributions are studied in  \cite{acemoglu_2011}. However, as indicated, fully Bayesian approaches are often intractable and therefore \emph{locally} Bayesian learning, also known as non-Bayesian learning, becomes necessary and helpful, as already demonstrated in various references including \cite{lalitha_2018,matta2016,bala98,bala2001,demarzo2003,epstein2008non,epstein2010non,golub2010,jadbabaie_2012,zhao_2012}. In this strategy, agents first update their beliefs via Bayesian updates based on their personal observations. Then, they average their immediate neighbors' beliefs with consensus \cite{degroot1974,jadbabaie_2012}, diffusion \cite{Sayed14,zhao_2012,salami2017} or gossip \cite{dimakis2010,shahrampour_2013} updates. These algorithms differ in which beliefs are combined. For instance, consensus algorithms are based on combining one's own updated belief with previous beliefs from neighbors. In this scheme, all agents need to have positive self-reliance for truth learning \cite{jadbabaie_2012}. Diffusion, on the other hand, combines one's own updated belief with the updated beliefs from neighbors. In this case, it is sufficient that at least one agent has positive self-reliance for truth learning \cite{zhao_2012}.

Social learning algorithms can also be classified based on the way beliefs are combined. Most common ways of aggregation are AA \cite{jadbabaie_2012,zhao_2012,jadbabaie2013information}, a.k.a. linear opinion pooling, and GA \cite{lalitha_2018,nedic_2017,inan2022social,bordignon2021adaptive}, a.k.a. logarithmic opinion pooling although there are variants, e.g., min-rule \cite{mitra2021minrule}. In AA, the combined belief is a convex combination of neighboring beliefs; while in GA, logarithms of the beliefs are combined linearly. GA is motivated by the linear combination of log-likelihood ratios in distributed detection theory.  Almost sure truth learning under both AA and GA are established in \cite{jadbabaie_2012,zhao_2012} and \cite{lalitha_2018,nedic_2017,inan2022social,toghani2021communication,bordignon2021adaptive}, respectively. An upper bound for the asymptotic learning rate of AA is given in \cite{jadbabaie2013information}, while the works \cite{lalitha_2018,nedic_2017,inan2022social,toghani2021communication} provide the asymptotic learning rate of social learning for GA. Even though one can conclude from these works that GA is faster than AA for learning the truth, the performance difference is still not clearly established in the literature. Knowing how much gain/loss in learning rate a distributed system would get if GA or AA is used, and how this gap is affected by the graph topology, are useful indicators for many applications. This work contributes to answering these questions. \\ 

\subsubsection{Applications}
In this section, we present some applications where the distinction between AA and GA is helpful. The reader may refer to \cite{koliander2022fusion} for other examples. 

\begin{itemize}
    \item \textit{Multi-sensor signal processing}: Fusion of information remains a challenge in multi-sensor signal processing. For example, this is relevant for Internet of Things (IoT) applications. Generally, controlling multiple sensor nodes by utilizing the \emph{global} observation model might be too complex. To alleviate this problem, opinion pooling strategies like AA and GA become useful when aggregating locally processed information. One possible application is multi-target tracking where agents aim to track an unknown number of objects. GA is used in \cite{uney2013,da2020}, whereas AA is used in \cite{yu2016,li2020bernoulli}. The works \cite{li2019second,li2021some,gao2020} include comparisons of AA and GA for this application. However, they are limited to one step fusion analysis and/or well-behaved distributions. Even though the present work is not explicitly targeting multi-object tracking, general results we obtain here can help elucidate the distinction between AA and GA in this area as well. 
    \item \textit{Machine learning}:
Merging probability vectors is also applicable for situations that require learning from data. For example, the work \cite{bordignon2021learning} proposes a social machine learning strategy where agents combine soft decisions, i.e., they combine belief vectors, in order to solve sequential classification problems. This approach is shown to outperform standard ensemble learning techniques such as AdaBoost \cite{Freund99ashort}. Our results can be extended to this setting too.\\
\end{itemize}

\vspace{-2em}

\subsection{Novelty and Contributions}

Since AA and GA have different attributes that can be useful for different applications (see Section \ref{sec:pooling_functions}), a comparison of their performance is crucial for distributed inference. The comparisons in the literature are limited in the sense that: ($i$) only one step of fusion is studied\cite{li2019second,li2021some,gao2020}, ($ii$) analysis is restricted to well-behaved distributions like Gaussian or Poisson distributions\cite{li2019second,chang2010,gao2020}, or ($iii$) a preliminary comparison is provided without detailed explanation \cite{lalitha_2018,dedecius2018}. In contrast, in this work, we study the {\em repeated} application of AA and GA in a canonical distributed inference problem and {\em without} confining to specific distributions. Moreover, we analyze the performance gap between AA and GA in detail. In particular, for social learning, we have the following results.

\begin{itemize}
    \item In Theorem \ref{thm:existence}, we prove that the agents learn the truth {\em exponentially fast} with the AA fusion rule under the standard diffusion algorithm \cite{Sayed14}. Furthermore, the decay rate of beliefs over a wrong hypothesis --- also called the learning rate --- is {\em constant} and does not depend on the agent.
    \item In Lemma \ref{lem:minrowsum} of Section \ref{sec:bounds}, we provide upper and lower bounds for this decay rate using superadditive and subadditive functions on matrices. Also, an interesting ``inept agent phenomenon'' is discovered by using an appropriate superadditive function.
    \item We provide a variational lower bound on the gap between the decay rates of AA and GA in Theorem \ref{thm:gap}. The bound involves the Dobrushin coefficient \cite[Chapter 2.7]{krishnamurthy_2016} as a network connectivity parameter. If the network is geometrically ergodic \cite{krishnamurthy_2016} (for which the Dobrushin coefficient is strictly smaller than 1), then the gap is zero if, and only if, the agents observe exactly the same data. Otherwise GA performs better in terms of learning rate.
    \item For the special case of rank-one combination matrices, which is equivalent to architectures with fusion center in terms of performance, the exact decay rate of AA in closed form is given in Proposition~\ref{prop:rankone}.
    \item For exchangeable networks, where no permutation of data across agents can change the dynamics, we also provide a closed form expression for the gap between the decay rates of AA and GA in Theorem \ref{thm:exchangeable}. \\
\end{itemize} 
\vspace{-1em}
\noindent\textbf{Notation:} We work in the probability space $(\Omega,\cF,\mathbb{P})$. We denote an element of $\Omega$ with the letter $\omega$. The random variables are all $\cF$-measurable and are denoted with boldface letters, e.g., $\bL_i$. Sets and events are represented with script-style letters (e.g., $\cA$). $|\cA|$ denotes the cardinality of set $\cA$. An almost sure event refers to an event $\cA \in \cF$ with probability 1. \( D(\cdot||\cdot) \) denotes the KL divergence. $\mathds{1}_{K}$ is the all-ones column vector of dimension $K$.
\vspace{-0.2em}
\section{\textbf{Background and Problem Formulation}}\label{sec:background}
\subsection{Inference Problem}
\begin{figure}[]
	\centering
	\includegraphics[width=2.5in]{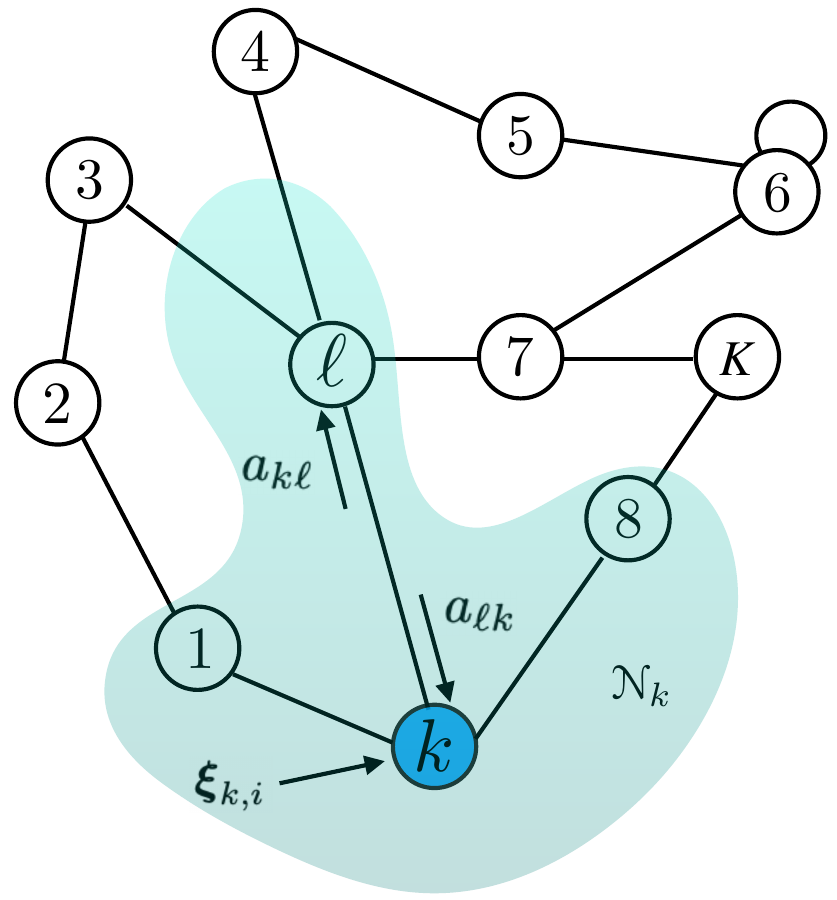}
	\caption{An example network diagram.}
	\label{fig:network_problem_formulation}
\end{figure}
We consider \( K \) agents collaboratively seeking to learn the true state of nature \( \theta^\circ\) from among a finite set of \( H \) hypotheses, \( \Theta = \{1,2,\dots,H \}\). At time instant \( i \), the confidence that agent \( k \) has that ``hypothesis \( \theta \) is the true hypothesis'' is denoted by \( \bmu_{k,i} (\theta) \). In this notation, the symbol $\bmu_{k,i}$ is a probability vector over all hypotheses. Agent \( k \) observes a partially informative and private signal \( \bxi_{k,i} \), which is distributed according to some known marginal likelihood, \( L_k(\cdot| \theta^\circ)\). This observation process is independently and identically distributed (i.i.d.) over time \( i \). However, independence across agents is not required. Agents know their own marginal likelihoods but do not know the likelihoods of other agents. In order to avoid pathological cases, we introduce the following assumption.
\begin{assumption}[Finite KL divergences] For each agent \( k \) and all hypotheses \( \theta \in \Theta \),
\begin{align}
    D ( L_k (\cdot | \theta^\circ ) ||  L_k (\cdot | \theta )) < \infty.
\end{align}\putqed
\end{assumption}
In addition, we assume for each wrong hypothesis $\theta \neq \theta^\circ $ that there exists at least one clear-sighted agent who can distinguish this hypothesis from the true hypothesis. This is necessary for learning the truth.
\begin{assumption}[Global identifiability]
For each wrong hypothesis \( \theta \in \Theta \setminus \{ \theta^\circ\} \), there exists at least one agent \( k \) with
\begin{align}
    D ( L_k (\cdot | \theta^\circ ) ||  L_k (\cdot | \theta )) > 0.
\end{align}\putqed
\end{assumption}
Observe that this assumption does not require \emph{local} identifiability, which is the ability of inferring \( \theta^\circ \) without any cooperation. As a result, agents must exchange information in order to uniquely identify the true hypothesis in general.

\subsection{Agent Topology}
There are two architectures of interest. \vspace{-1em} \\
\subsubsection{Decentralized peer-to-peer networks}
For this architecture, we consider a strongly-connected graph \cite{Sayed14} governing the communication topology. In other words, for any agent pair \( (\ell,k)\) there exists a path in between and linking them and, moreover, there exists at least one agent in the network with positive self-reliance, denoted by $a_{kk}>0$ for some $k$. This allows the information to diffuse across the network thoroughly. We associate non-negative combination coefficients  $0\leq a_{\ell k}\leq 1$ with each link from $\ell$ to $k$ and collect these weights into the combination matrix. The strong-connectedness of the graph translates into $A$ being a primitive matrix \cite{Sayed14}. Each \( a_{\ell k} \) represents the weight agent \( k \) assigns to the information received from agent \( \ell \). It is non-zero if, and only if, \( \ell \in \mathcal{N}_k\), i.e., agent \( \ell \) is an immediate neighbor of agent \( k \). The matrix \( A \) is also left-stochastic. This means that the entries on each of its columns add up to 1. Such matrices have a unique eigenvalue at 1, and we denote the corresponding eigenvector $\pi$ whose entries are normalized to add up to one:
\begin{align}
   \mathds{1}_{K}^{\T}\pi=1 ,\quad A\pi=\pi, \quad \mathds{1}_{K}^{\T}A=\mathds{1}_{K}^{\T}.
\end{align} \\
The vector \( \pi \) is the Perron vector of \( A \), and by the Perron-Frobenius theorem, all entries of $\pi$ are positive \cite{Sayed14}. \vspace{-1em}\\
\subsubsection{Architectures with fusion center}
Agents communicate with a central node instead of communicating with other agents. The central node aggregates the information received from peripheral agents by taking a weighted average. Note that the performance of this architecture is equivalent to the performance of a fully-connected network topology with a rank-one combination matrix \( A \). Thus, in the following, we focus on general combination matrices \( A \) described for decentralized networks. For the special case of $A$ being rank one, we obtain additional results on top of the general case --- see Proposition \ref{prop:rankone} in Section \ref{sec:rankone}.

\subsection{Pooling Functions}\label{sec:pooling_functions}
We compare linear (AA) \cite{stone1961} and log-linear (GA) \cite{genest1984} fusion of probability vectors in this work. The AA and GA rules in the context of social learning are presented in Section \ref{sec:social_learning}. Here, we comment briefly on some features of AA and GA. In the current work, we do not focus on identifying various features of AA and GA. The reader may refer to \cite{koliander2022fusion} for other useful characteristics. 

While GA is externally Bayesian, that is, Bayesian updates on the beliefs and the aggregation step are commutative, it nevertheless gives agents a \emph{veto power}. That is, if one agent has zero belief on some hypothesis, then the composite belief will also be zero on that same hypothesis regardless of the beliefs by the other agents. The AA rule, on the other hand, does not give this much power to individual agents and hence can be more robust against adversarial attacks. Moreover, in the case of Gaussian beliefs over continuous hidden variables, repeated application of GA preserves Gaussianity and is related to the covariance-intersection method \cite{hu2012}. In contrast, the properties of AA-combined distributions diverge from the Gaussian distributions. 

\subsection{Diffusion Social Learning based on AA and GA}\label{sec:social_learning}
Under diffusion social learning, at each iteration $i$, agents first update their beliefs using a \emph{local} Bayesian rule based on their private observations and compute the following  intermediate beliefs:
\begin{align}\label{eq:dif_adapt_step}
 \bpsi_{k,i} (\theta) &= \frac{L_k(\bxi_{k,i} | \theta)\bmu_{k,i-1} (\theta)}{\sum_{\theta^\prime}L_k(\bxi_{k,i} | \theta^\prime)\bmu_{k,i-1} (\theta^\prime)}.
\end{align}
The next step differs depending on whether AA or GA fusion is used. Under AA, each agent updates its belief vector by computing a weighted arithmetic average of its neighbors' intermediate beliefs. More precisely, for agent $k$, the linear fusion is given by the following rule:
\begin{align}\label{eq:linear_fusion}
   \bmu_{k,i}(\theta) &= \sum_{\ell \in \cN_k}{a_{\ell k}\bpsi_{\ell,i}(\theta)}. \qquad\text{(AA-Diffusion)}
\end{align}
Since the $\{a_{\ell k}\}$ are positive coefficients that add up to one, the resulting vector $\bm{\mu}_{k,i}$ will be a belief vector with its entries adding up to one. Under GA, on the other hand, agents average the intermediate beliefs of their neighbors in a geometric manner and also perform a normalization step to ensure that the resulting vector continues to be a probability vector (with entries adding up to one):
\begin{equation}\label{eq:geometric_fusion}
   \bmu_{k,i}(\theta) \!=\!\frac{\text{exp} \{ \sum_{\ell \in \mathcal{N}_k } a_{\ell k} \log \bpsi_{\ell,i} (\theta )\}}{\sum_{\theta'}\text{exp} \{ \sum_{\ell \in \mathcal{N}_k } a_{\ell k} \log \bpsi_{\ell,i} (\theta')\}} \:\:\text{(GA-Diffusion)}
\end{equation} \vspace{-1em}\\

\noindent \textbf{Remark:} Observe that the support of the averaged belief vectors is the union of the supports in the AA case, while it is the intersection of the supports in the GA case. Hence, in GA, all agents need to have positive initial beliefs for all hypotheses in order not to discard any hypothesis. In contrast, it is enough for at least one agent to have a positive initial belief for any hypothesis in order not to get discarded in AA. \qed \vspace{-1em}\\


\noindent \textbf{Remark:} In the consensus algorithm, agents combine their intermediate beliefs with the previous beliefs of the neighbors instead. Compared with \eqref{eq:linear_fusion}, the combination step of the AA-Consensus is given by:
\begin{align}\label{eq:linear_fusion_consensus}
   \bmu_{k,i}(\theta) = {a_{k k}\bpsi_{k,i}(\theta)} \! + \! \sum_{\ell \in \cN_k\setminus \{k\}}{a_{\ell k} \bmu_{\ell,i}(\theta)} \qquad \quad  \text{(AA-Consensus)}
\end{align}
Similarly, compared with \eqref{eq:geometric_fusion}, the combination step of the GA-Consensus is given by:
\begin{equation}\label{eq:geometric_fusion_consensus}
   \bmu_{k,i}(\theta) = \dfrac{\text{exp} \Big \{ a_{k k} \log \bpsi_{k,i} (\theta ) + \sum_{\ell \in \cN_k\setminus \{k\}} a_{\ell k} \log \bmu_{\ell,i} (\theta ) \Big \}}{\sum_{\theta'}\text{exp} \Big \{ a_{k k} \log \bpsi_{k,i} (\theta' ) + \sum_{\ell \in \cN_k\setminus \{k\}} a_{\ell k} \log \bmu_{\ell,i} (\theta' )\Big \}} \qquad \text{(GA-Consensus)}
\end{equation}
\putqed


\section{\textbf{Existence of Asymptotic Decay Rates}}
Recall that the true hypothesis is fixed and is denoted by $\theta^\circ$. To measure how fast the beliefs converge to the truth (i.e., how fast $\bm{\mu}_{k,i}(\theta^\circ)\rightarrow 1$ and $\bm{\mu}_{k,i}(\theta\neq \theta^\circ)\rightarrow 0$), for each $k$, it is sufficient to study the asymptotic behavior of agent $k$'s belief on a false hypothesis $\theta \neq \theta^\circ$. More precisely, we will study the exponential decay rates of these beliefs, defined as
\begin{equation}\label{eq:asym_learning_rate_def}
    \brho_{k}(\theta) \triangleq -\limsup_{i \to \infty} \frac 1 i \log \bmu_{k,i}(\theta)
\end{equation}
although there are variations. For instance, in \cite{jadbabaie2013information}, the agent regrets are defined as the total variation distance to the truth vector $e_{\theta^\circ}$, whose $\theta^\circ$th element is one and others are zero,
\begin{equation}
   \mathbf{d}_{k,i} \triangleq \frac 1 2\|\bmu_{k,i}(\theta) - e_{\theta^\circ}\|_1,
\end{equation}
and the asymptotic behavior of the average regret on the network is analyzed, i.e.,
\begin{align}\label{eq:TV_regret}
\brho \triangleq \liminf_{i \to \infty}\frac 1 i \bigg|\log \bigg(\sum_{k = 1}^K\mathbf{d}_{k,i}\bigg)\bigg|.
\end{align}
Although they do seem different, it can be shown that for both definitions, the quantities of interest are $\brho_{k}(\theta)$'s --- see Appendix~\ref{app:A} for the proof. Hence, we focus on \eqref{eq:asym_learning_rate_def} for the rest of this work. Note that some definitions of exponential convergence require $\bmu_{k,i}(\theta) \leq C \exp\{-\rho i \}$ where $C$ is a uniformly bounded constant. In this work, we adhere to the definition from the existing social learning literature, where exponential convergence refers to the presence of the exponent, and does not the require a uniform bound on $C$.

In \cite{jadbabaie2013information}, the authors obtained upper and lower bounds for $\brho$ for the \emph{consensus} algorithm. In this section, we show that under the AA diffusion rule \eqref{eq:linear_fusion} that {\em exact} limits exist in \eqref{eq:asym_learning_rate_def} with $\limsup$ replaced by $\lim$ and that these limits are independent of the agent index $k$. That is, we show that 
\begin{equation}
\brho^{(\AA)}(\theta) \triangleq \lim_{i \to \infty} -\frac 1 i \log \bmu_{k,i}(\theta)
\end{equation}
exists almost surely and does not depend on $k$. This means that  the beliefs on false hypotheses decay exponentially at the same rate for all agents. In comparison to GA-diffusion, the analysis for AA-diffusion is much more involved. This is because GA-diffusion is amenable to an analysis that studies the evolution of log-belief ratios, and can benefit from an application of the strong law of large numbers.
\subsection{Brief Analysis of GA-diffusion}
We start with the GA-diffusion rule \eqref{eq:geometric_fusion}. Following the works \cite{lalitha_2018,nedic_2017,inan2022social}, we provide a summary of the studies done for finding the decay rates. It is convenient to introduce the vector of log-belief ratios, $\bS_{i}(\theta) \triangleq [\bS_{1,i}(\theta), \dots, \bS_{K,i}(\theta)]^{\T}$ where
\begin{equation}
\bS_{k,i}(\theta) \triangleq \log\frac{\bmu_{k,i}(\theta^{\circ})}{\bmu_{k,i}(\theta)}.
\end{equation}
It can be easily verified from \eqref{eq:dif_adapt_step} and \eqref{eq:geometric_fusion} that the vector $\bS_{i}(\theta)$ evolves according to the linear stochastic system
\begin{align}\label{eq:LLF_evolution}
   \bS_{i}(\theta) =  A^{\T}(\bS_{i-1}(\theta) + \bL_{i}(\theta))
\end{align}
where $\bL_{i}(\theta) \triangleq [\bL_{1,i}(\theta), \dots, \bL_{K,i}(\theta)]^{\T}$ is the vector of log-likelihood ratios (LLRs):
\begin{equation}
\bL_{k,i}(\theta) \triangleq \log \frac{L_k(\bxi_{k,i} | \theta^\circ)}{L_k(\bxi_{k,i} | \theta)}.
\end{equation}
Repeated ($i$-fold) application of \eqref{eq:LLF_evolution} yields
\begin{align}\label{eq:llr_slln}
    \frac 1 i \bS_{i}(\theta) =  \frac 1 i\sum_{j = 1}^{i}  (A^{\T})^{i-j+1}\bL_{j}(\theta) +  \frac 1 i (A^{\T})^i\bS_{0}^{\T}(\theta).
\end{align}
Now recall that ($i$) $\bL_{i}(\theta)$'s are i.i.d. random vectors, ($ii$) $A^i~\to~\pi\mathds{1}_{K}^{\T}$, with $\pi$ being the Perron vector of $A$, and ($iii$) 
\begin{align}
\E[\bL_{k,i}(\theta)] &= \E\bigg[\log \frac{L_k(\bxi_{k,i} | \theta^\circ)}{L_k(\bxi_{k,i} | \theta)}\bigg] \notag \\&= D(L_k(.|\theta^{\circ}) || L_k(.|\theta)) < \infty
\end{align}
by Assumption 1. Using these conditions we first note that the first term on the right-hand side of \eqref{eq:llr_slln} tends to
\begin{equation}
    \frac 1 i \sum_{j = 1}^{i}  (\mathds{1}_K\pi^{\T})\bL_{j}(\theta),\quad i\to\infty.
\end{equation}
Next, an application of the strong law of large numbers gives almost surely
\begin{align}\label{eq:LLF_rate}
\lim_{i \to \infty}\frac 1 i \bS_{k,i}(\theta) = \sum_{\ell=1}^K \pi_\ell \E[\bL_{\ell,i}(\theta)] > 0
\end{align}
which does not depend on $k$. Result \eqref{eq:LLF_rate} readily implies that for all $\theta \neq \theta^\circ$, $\frac{\bmu_{k,i}(\theta)}{\bmu_{k,i}(\theta^\circ)} \to 0$ and since $\sum_{\theta} \bmu_{k,i}(\theta) = 1$, then $\bmu_{k,i}(\theta^\circ)\to 1$ almost surely.
Hence, with GA-diffusion, the decay rates are the constants given by
\begin{align}\label{eq:rho_def}
    \rho^{(\GA)}(\theta) &\triangleq\lim_{i \to \infty} \frac 1 i \bS_{k,i}(\theta) \notag \\ &= \lim_{i \to \infty} \frac 1 i \log \frac{\bmu_{k,i}(\theta^\circ)}{\bmu_{k,i}(\theta)}\notag \\ &\stackrel{(a)}{=} -\lim_{i \to \infty} \frac 1 i \log \bmu_{k,i}(\theta) \notag \\
      &= \sum_{k=1}^K \pi_k D(L_k(.|\theta^{\circ}) || L_k(.|\theta))
\end{align}
where $(a)$ follows from the fact that $\bmu_{k,i}(\theta^\circ)$ tends to 1 almost surely. Observe that the decay rates are characterized by a weighted average of the KL divergences of the agents --- it is known that these entities reflect the inference capacity of an agent for a hypothesis testing problem \cite{chernoff52}.

\subsection{Asymptotic Decay Rate of AA-Diffusion}
Unfortunately, a similar analysis is not possible for AA-diffusion --- if we attempt to study $\log\bmu_{k,i}(\theta)$ directly, we end up with
\begin{align}\label{eq:intractable2}
    \log\bmu_{k,i}(\theta) = \log\bigg(\sum_{\ell \in \cN_k}{a_{\ell k}\bpsi_{\ell,i}(\theta)}\bigg),
\end{align}
 which does not provide a simple-to-analyze dynamical system like we had with \eqref{eq:LLF_evolution}. Thus, we need to resort to different methods in the following. The authors of \cite{jadbabaie2013information} approached the problem of finding $\brho^{(\AA)}$ by linearizing the dynamical system \eqref{eq:intractable2}. With this method, they were able to lower bound $\brho^{(\AA)}$ by the Lyapunov exponent of the linearized version. We take a different approach by constructing extremal processes that bound $\bmu_{k,i}(\theta)$. \vspace{-1em}\\
 
 \subsubsection{Constructing the Extremal Process}
 Recall the evolution of the process $\{\bmu_{k,i}\}$ given by \eqref{eq:dif_adapt_step} and \eqref{eq:linear_fusion} under AA-diffusion. We wish to simplify the analysis by obtaining an extremal process $\{\bnu_{k,i}\}$, which eventually remains above $\{\bmu_{k,i}\}$ with probability 1. Studying $\{\bnu_{k,i}\}$ will then lead to a bound on the decay rate for AA-diffusion. With this aim, we first recall the following theorem.
 \begin{theorem}[Truth learning \cite{zhao_2012}]\label{thm:Ali}Under Assumptions 1 and 2, when AA-diffusion is executed, all agents learn the truth, i.e., for each agent \( k \) we have almost surely
 \begin{equation}
 \lim_{i \to \infty} \bmu_{k,i}(\theta^\circ) = 1.
 \end{equation}\putqed
 \end{theorem}

 Recall that our setting lies in the probability space $(\Omega, \cF, \mathbb{P})$, where $\Omega$ represents the space of all data sequence realizations over time $\omega \in \Omega$, $\cF$ represents the $\sigma$-field generated by the sequence of data, and $\mathbb{P}$ represents the probability measure over sample paths $\omega \in \Omega$. In light of Theorem \ref{thm:Ali}, if we set $\epsilon > 0$ and define the event that all $\bmu_{k,i}(\theta^\circ)$'s lie above $1-\epsilon$ eventually as
 \begin{equation}\label{eq:set_G}
 \cG(\epsilon) \triangleq \{\omega \in \Omega: \exists i_0\  \bmu_{k,i}(\theta^\circ)\geq 1-\epsilon,\ \forall i > i_0,\forall k\},
 \end{equation}
 observe that $\bP{\cG(\epsilon)} = 1$ as a consequence of Theorem \ref{thm:Ali}. Note that $\cG(\epsilon)$ is also interpreted as the event that there exists an $\bm{i}_0(\omega)$ --- which is a random variable as it depends on $\omega$ --- such that for all agents the true beliefs remain greater than $1-\epsilon$ after $\bm{i}_0^{\text{th}}$ iteration. We now restrict ourselves to $\cG(\epsilon)$ and study the evolution of $\{\bmu_{k,i}\}$ under such restriction.

	Consider a false hypothesis $\theta \neq \theta^\circ$. We study the pathwise trajectories of $\{\bmu_{k,i}(\theta)\}$ for an outcome in $\cG(\epsilon)$, that is, we pick an $\omega \in \cG(\epsilon)$. According to the definition of $\cG(\epsilon)$, there exists an $\bm{i}_0(\omega)$ such that for all $i > \bm{i}_0(\omega)$:
	\begin{align}
	\label{eq:linear_update}\bmu_{k,i}(\theta) &=  \sum_{\ell \in \cN_k} a_{\ell k} \frac{\bmu_{\ell,i-1}(\theta)L_{\ell}(\bxi_{\ell, i}|\theta)}{\sum_{\theta'} \bmu_{\ell, i-1}(\theta')L_{\ell}(\bxi_{\ell, i}|\theta')}\\
	\label{eq:upper_process}
	&\leq \sum_{\ell \in \cN_k} a_{\ell k} \frac{\bmu_{\ell,i-1}(\theta)L_{\ell}(\bxi_{\ell, i}|\theta)}{(1-\epsilon)L_{\ell}(\bxi_{\ell, i}|\theta^\circ)}.
	\end{align}
	Let \begin{align}
	    \br_{k,i}(\theta) \triangleq \frac{L_{k}(\bxi_{k, i}|\theta)}{L_{k}(\bxi_{k, i}|\theta^\circ)}
	\end{align} be the likelihood ratio of the freshly observed data by agent $k$ at time $i$, between the hypotheses $\theta$ and $\theta^\circ$. We define the extremal process $\{\bnu_{k,i}\}$ as the process that evolves according to
	\begin{align}\label{eq:nu_process}
	    \bnu_{k,i}(\theta) \triangleq (1-\epsilon)^{-1}\sum_{\ell \in \cN_k} a_{\ell k}\bnu_{\ell,i-1}(\theta)\br_{\ell,i}(\theta)
	\end{align}
	for all $i \geq \bm{i}_0(\omega)$ and with $\bnu_{k,\bm{i}_0}(\theta) = \bmu_{k,\bm{i}_0}(\theta)$. Comparing \eqref{eq:nu_process} with \eqref{eq:upper_process}, we see that $\{\bmu_{k,i}\}$ is upper bounded by $\{\bnu_{k,i}\}$ for all $i\geq \bm{i}_0(\omega)$. \vspace{-1em}\\
	
	\noindent\textbf{Remark:} For the rest of the work, we fix a false hypothesis $\theta \neq \theta^\circ$ and omit $\theta$ as an argument for brevity. This is because the analysis is the same for all $\bmu_{k,i}(\theta)$. We write the dependencies whenever we want to emphasize them. \putqed \vspace{-1em}\\
	
	The transition from $\bnu_{k,i-1}$ to $\bnu_{k,i}$ given by \eqref{eq:nu_process} is a random linear transform. Let $\bnu_i \triangleq [\bnu_{1,i},\dots,\bnu_{K,i}]^{\T}$ and define the diagonal \(K \times K\) random diagonal matrices $\bR_i$ with their $k^{\text{th}}$ diagonal element being $\br_{k,i}$. Then for all $\omega \in \cG(\epsilon)$, expression \eqref{eq:nu_process} leads to the following vector relation:
	\begin{align}\label{eq:matrix_form}
	    \bnu_{i} = (1-\epsilon)^{-1}(A^{\T}\bR_i) \bnu_{i-1}, \quad \forall i > \bm{i}_0(\omega).
	\end{align}
    Equation \eqref{eq:matrix_form} highlights that the asymptotic behavior of the random matrix product \begin{equation}\label{eq:Yi_def}
    \bY_i \triangleq \prod_{j = 1}^i (A^{\T}\bR_j)\end{equation}
    plays an important role. Observe that since $\{\bR_i\}$ is a stationary sequence, the asymptotic behavior remains unchanged under any time shift. Hence, starting the product from $j=1$ in equation \eqref{eq:Yi_def} is without loss of generality. We now turn our attention to the analysis of the random matrices $\{\bY_i\}$. \vspace{-1em}\\
\subsubsection{Asymptotic Behavior of $\{{\it \bY}_i\}$}
    The asymptotic behavior of random matrix products is an important and challenging problem with a long history, which includes the preliminary study \cite{Bellman}, and later the seminal work \cite{Furstenberg}. Although some results exist under certain assumptions on the structure of the random matrices to be multiplied, the problem remains open in general. One difficulty is that the non-commutativity of matrices under multiplication prevents the use of well-known convergence theorems such as the law of large numbers. Extending the result of \cite{Furstenberg}, reference \cite{Kingman} studied the more general case of ``subadditive processes'', and derived, under some fairly general conditions, an ergodic theorem known as Kingman's subadditive ergodic theorem \cite{Kingman}. The analysis of random matrix products turns out to be a special case of this result.
    \begin{theorem}[Kingman's subadditive ergodic theorem\cite{Kingman}]\label{thm:Kingman}
    Consider a stationary sequence of random matrices $\{\bX_i\}$ and suppose that the elements of $\bX_i$'s are positive and that their logarithms have finite expectations. Let $\bY_i \triangleq \prod_{j=1}^i \bX_j$. Then, the limit
\begin{equation}
\bm{\gamma} =  \lim _{i \rightarrow \infty} \frac 1 i \log \left[\bY_{i}\right]_{\ell k}
\end{equation}
exists and is finite almost surely and in the mean, and does not depend on $\ell$ or $k$. Furthermore, it holds that
\begin{equation}
\E[\bm{\gamma}] = \lim _{i \rightarrow \infty} \frac 1 i \E[\log \left[\bY_{i}\right]_{\ell k}].
\end{equation}
    \end{theorem}\putqed

We now adapt the above theorem to our setting. First of all, observe that the $\{\bR_i\}$ is an i.i.d. sequence, and therefore the $\{A^{\T}\bR_i\}$ is also i.i.d., and hence, stationary. Note that we cannot simply replace the $\bX_i$'s with $A^{\T}\bR_i$'s because the matrix $A^{\T}$ need not have all positive entries. However, $A$ is a primitive matrix and thus there must exist some $n \geq 1$ such that every entry of $A^n$ is strictly positive \cite{Sayed14} and \cite[Chapter 8]{horn2012matrix}. Using this observation, we arrive at the following corollary.
\begin{corollary}[Limit of the random matrix product]\label{cor:existence_gamma}
Consider $\bY_i = \prod_{j = 1}^i (A^{\T}\bR_j)$. Under Assumptions 1 and 2, the finite limit 
\begin{align}
\gamma=\lim _{i \rightarrow \infty} \frac 1 i \log \left[\bY_{i}\right]_{\ell k} = \lim _{i \rightarrow \infty} \frac 1 i \E[\log \left[\bY_{i}\right]_{\ell k}]
\end{align}
exists almost surely, is a constant, and does not depend on $\ell$ or $k$.
\end{corollary}
\begin{proof}
See Appendix \ref{app:existence_gamma}.
\end{proof}

Using Corollary \ref{cor:existence_gamma}, we are able to characterize the asymptotic decay rate for AA-diffusion and conclude this section.
\begin{theorem}[Asymptotic decay rate of AA-diffusion]\label{thm:existence} For any agent $k$, it holds almost surely:
\begin{align}\label{eq:asymp_rate}
  \brho_{k}^{(\AA)}(\theta) = -\lim_{i \to \infty} \frac 1 i \log \bmu_{k,i}(\theta) = -\gamma(\theta).
\end{align}
\end{theorem}
\begin{proof}
See Appendix \ref{app:existence}. 
\end{proof}

\noindent\textbf{Remark:} A weaker result obtained in \cite{jadbabaie2013information} ensures $\brho \geq -\gamma$, with $\brho$ defined in \eqref{eq:TV_regret}. It can be verified that this is an implication of Theorem \ref{thm:existence}.\putqed \vspace{-1em}\\

Theorem \ref{thm:existence} states that if AA-diffusion is executed, the beliefs on a false hypothesis $\theta$ decay exponentially almost surely, and the decay rate is constant and is the same among all agents. Note, however, that the decay rate may vary across $\theta$. This is because $\gamma = \gamma(\theta)$ is the limit pertaining to the i.i.d. products of the matrices $A^{\T}\bR_i(\theta)$, where $\bR_i(\theta)$ is defined in terms of the $\br_{k,i}(\theta)$ --- see \eqref{eq:nu_process}.

\section{\textbf{Bounds on the Asymptotic Decay Rate}}\label{sec:bounds}
It is stated in \cite{Kingman} that ``pride of place among the unsolved problems of subadditive ergodic theory must go to the calculation of the constant $\gamma$". To the best of our knowledge, no standard machinery exists to date for this end. Therefore, we make use of the special structure of the matrix $A^{\T}\bR_i$ to obtain bounds for $\gamma$ in Sections~\ref{sec:distributed_inf_case_bounds} and \ref{sec:special_cases_exact}. But before those, we first provide some simple upper and lower bounds that hold for the general cases.
\subsection{Bounds Based on Subadditivity}\label{sec:subadditive}
First, it is useful to discuss why products of random matrices are related to subadditive processes. Let $\|X\|$ denote any matrix norm that is submultiplicative, i.e., for any $X$ and $Y$
\begin{align}
    \|XY\| \leq \|X\|\|Y\|.
\end{align}
For our problem at hand, replacing $X$, $Y$ with $\bY_i^j$, $\bY_j^m$ whose general definition is given as
\begin{equation}\label{eq:matrix_product_main}
\bY_{i}^m \triangleq  \prod_{t = i+1}^m (A^{\T}\bR_t),
\end{equation}
and taking the logarithms of the both sides yield the subadditive relation
\begin{align}
    \log\|\bY_i^m\| \leq \log \|\bY_i^j\| + \log \|\bY_j^m\|.
\end{align}
A well-known property of subadditive functions \cite{Gallager} is that
\begin{align}\label{eq:subadditive}
    \lim_{i \to \infty} \frac 1 i \log\|\bY_i\| = \inf_i \frac 1 i \log\|\bY_i\|,
\end{align}
where \(\bY_i\) is defined in \eqref{eq:Yi_def}. In fact, this observation is the starting point of the work \cite{Furstenberg} on random matrix products. Now, consider the norm of matrices with non-negative entries: $\|X\|_1 \triangleq \max_{\ell} \sum_{k} [X]_{k\ell}$, which is submultiplicative. It is then imminent from Corollary \ref{cor:existence_gamma} and \eqref{eq:subadditive} that for any $j$
\begin{align}\label{eq:lemma2}
    \gamma &= \lim_{i \to \infty} \frac 1 i \E\big[\log[\bY_i]_{11}\big] \leq \lim_{i \to \infty} \frac 1 i \E\big[\log\|\bY_i\|_1\big] \notag \\ 
     &=\inf_i \frac 1 i \E\big[\log\|\bY_i\|_1\big] \leq \frac 1 j \E\big[\log\|\bY_j\|_1\big],
\end{align}
which yields an upper bound for $\gamma$.

For the lower bound, we aim to create a supermultiplicative process. Let $\|X\|_- \triangleq \min_{\ell} \sum_k [X]_{k \ell} $ be the minimum column sum of the matrix $X$. Note that this is not a norm. However, it is supermultiplicative for non-negative matrices, i.e., 
\begin{align}
    \|XY\|_- \geq \|X\|_-\|Y\|_-.
\end{align}
We then obtain the following result.
\begin{lemma}[Bounds based on subadditivity]\label{lem:minrowsum}
For any $i$, $j \geq 1$,
\begin{align}
   \frac 1 i \E\big[\log\|\bY_i\|_-\big] \leq \gamma \leq \frac 1 j \E\big[\log\|\bY_j\|_1\big].
\end{align}
\end{lemma}
\begin{proof}
The upper bound follows directly from \eqref{eq:lemma2}. For the lower bound, observe that $-\log\|\bY_i\|_-$ is a subadditive process. We then have
\begin{align}
    \sup_i \frac 1 i \E\big[\log\|\bY_i\|_-\big] & = \lim_{i \to \infty} \frac 1 i \E\big[\log\|\bY_i\|_-\big] \notag \\
    &= \E\bigg[\lim_{i \to \infty} \frac 1 i\log\|\bY_i\|_-\bigg] = \gamma.
\end{align}
where the interchange of the limit and expectation operations is due to the general form of the subadditive ergodic theorem --- see Appendix \ref{app:subadditive} for more detail. 
\end{proof} 
We point out that $\|\cdot\|_- $ and $\|\cdot\|_1$ can be replaced with any suitable functions $f$, $g$ on positive matrices that are super/submultiplicative respectively; and satisfy the conditions of the subadditive ergodic theorem. This way, one obtains lower and upper bounds for the rate of AA-diffusion with a method similar to what we did in Lemma \ref{lem:minrowsum}.

In light of the above observation, we remark that $\|.\|_-$ can also be replaced with an element of $\bY_i$, e.g., with any mapping $\bY_i\mapsto [\bY_i]_{k k}$, as long as it remains strictly positive for all $i$ --- this is to ensure that the logarithm remains finite. It can be verified that $[\bY_i]_{k k}$ is supermultiplicative as long as $\bY_i$'s are non-negative. Recall that at least one agent has a self-loop, which implies $a_{kk} > 0$ for some $k$. Without loss of generality, assume $k =1$. The self-loop assumption then ensures that $[\bY_i]_{11}$ remains strictly positive. Then, according to Lemma \ref{lem:minrowsum}, we have 
\begin{equation}
\E[\log [\bY_1]_{11}] = \log a_{11} + \E[\log{\br_{1,1}}] \leq \gamma
\end{equation}
which, by Theorem~\ref{thm:existence}, implies that
\begin{align}
\rho^{(\AA)} &\leq -\log a_{11} - \E[\log{\br_{1,1}}]\notag\\
&=-\log a_{11} + D(L_1(\cdot|\theta^\circ)||L_1(\cdot|\theta)).\label{eq:single_element}
\end{align}
The inequality \eqref{eq:single_element} gives rise to an interesting observation of the learning model under AA-diffusion. For a small $\delta >0$, suppose that $a_{11}=1-\delta$ is close to one and $D(L_1(\cdot|\theta^\circ)||L_1(\cdot\theta))\leq \delta$ is small. This suggests that agent 1 is highly self-confident despite limited learning abilities. For such special case, observe that
\begin{align}
\rho^{(\AA)} &\leq -\log (1-\delta) + \delta  \leq  \delta/(1-\delta) + \delta
\end{align}
is also small. Therefore, since all agents learn the truth at the same rate $\rho^{(\AA)} $ by Theorem \ref{thm:existence}, the inept and self-confident agent 1 drastically decreases the learning ability of the whole network. This phenomenon can be avoided under GA-diffusion: if the remaining agents do not trust agent 1, i.e., if $a_{1k}$'s are small, then the first element of the Perron vector, $\pi_1$, can be kept small as well. It is evident from \eqref{eq:rho_def} that such isolation of agent 1 from the network will preserve (and might even boost) the learning rate of the remaining agents. This observation is consistent with the numerical results in Section \ref{sec:numerical}.

\subsection{Bounds for the Special Case of Distributed Inference}\label{sec:distributed_inf_case_bounds}
In this section, we take advantage of the special structure of $\{\bY_i\}$ in our distributed setting and derive bounds for $\gamma$. To that end, observe first that
\begin{align}
[\bY_i]_{1k} = \br_{k,i}\sum_\ell [\bY_{i-1}]_{1\ell} a_{k \ell}.
\end{align}
This implies that
\begin{align}
\log[\bY_i]_{1k} = \log\br_{k,i}+\log\bigg(\sum_\ell [\bY_{i-1}]_{1\ell} a_{k \ell}\bigg),
\end{align}
and averaging over the network with weights $\pi_k$ yields that
\begin{equation}
\sum_k\pi_k\log\frac{[\bY_i]_{1k}}{\pi_k} = \sum_k\pi_k\log\br_{k,i} +\sum_k\pi_k\log\frac{\sum_\ell [\bY_{i-1}]_{1\ell} a_{k \ell}}{\pi_k}.
\end{equation}
We subtract $\sum_k\pi_k\log\dfrac{[\bY_{i-1}]_{1k}}{\pi_k}$ from both sides of this equation to obtain
\begin{align}
\sum_k&\pi_k\log\frac{[\bY_i]_{1k}}{\pi_k}-\sum_k\pi_k\log\frac{[\bY_{i-1}]_{1k}}{\pi_k}\notag\\\label{eq:difference}
&= \sum_k\pi_k\log\br_{k,i}+\sum_k\pi_k\log\frac{\sum_\ell [\bY_{i-1}]_{1\ell} a_{k \ell}}{[\bY_{i-1}]_{1k}}.
\end{align}
If we take the time average of \eqref{eq:difference} from $j=2$ to $i$, the left-hand side becomes a telescoping sum where the intermediate terms cancel each other, and we arrive at the following relation: 
\begin{align}
&\frac 1 i \sum_k\pi_k\log\frac{[\bY_i]_{1k}}{\pi_k}-\frac 1 i\sum_k\pi_k\log\frac{[\bY_1]_{1k}}{\pi_k}\notag\\
&= \frac 1 i \sum_{j=2}^i\big(\sum_k\pi_k\log\br_{k,j}\big)+\frac 1 i \sum_{j=2}^i\sum_k\pi_k\log\frac{\sum_\ell [\bY_{j-1}]_{1\ell} a_{k \ell}}{[\bY_{j-1}]_{1k}}.
\end{align}
Assume $[\bY_1]_{1k} > 0$ for simplicity. Note that there is no loss of generality here since the strong connectivity assumption on the network ensures the primitiveness of the combination matrix $A$ which in turn ensures $[\bY_1]_{1k}$ to be strictly positive eventually--- see Appendix~\ref{app:existence_gamma}. The left-hand side tends to $\gamma$ by Theorem \ref{thm:existence}, and since $\br_{k,i}$'s are i.i.d., the first term on the right-hand side tends to its mean by the law of large numbers and is equal to the negative of the decay rate of GA-diffusion --- see \eqref{eq:LLF_rate}, \eqref{eq:rho_def}. In other words,
\begin{align}
\gamma = \sum_k\pi_k\E[\log\br_{k,j}]+\lim_{i \to \infty}\frac 1 i \sum_{j=2}^i\sum_k\pi_k\log\frac{\sum_\ell [\bY_{j-1}]_{1\ell} a_{k \ell}}{[\bY_{j-1}]_{1k}}
\end{align}
and accordingly,
\begin{align}
\rho^{{(\GA)}} - \rho^{(\AA)} \!= \!\lim_{i \to \infty}\frac 1 i \sum_{j=2}^i\sum_k\pi_k\log\frac{\sum_\ell [\bY_{j-1}]_{1\ell} a_{k \ell}}{[\bY_{j-1}]_{1k}}
\end{align}
The above equation quantifies the gap between the decay rates. The gap turns out to be the difference of two KL divergences averaged over time. To see this, let $\bu$ denote the probability vector obtained by normalization of the first row of $\bY$, i.e.,
\begin{align}
    [\bu_i]_k \triangleq \frac{[\bY_i]_{1k}}{\sum_\ell [\bY_i]_{1\ell}}.
\end{align}
Then, it follows that
\begin{equation}\label{eq:gamma_exact}
\rho^{{(\GA)}} \!- \rho^{(\AA)} \!=\!\!\lim_{i \to \infty}\frac 1 i \!\sum_{j=2}^i\!\!\big[D(\pi||\bu_{j-1}) - D(\pi||A\bu_{j-1})\big]
\end{equation}
Equation \eqref{eq:gamma_exact} has an interesting interpretation. Note that $A\bu_{j-1}$ is another probability vector that is obtained by passing $\bu_{j-1}$ through the Markov matrix $A$. Furthermore, $A\pi = \pi$ is the unique invariant distribution of the kernel $A$ and thus the difference term in \eqref{eq:gamma_exact} is equal to
\begin{align}
  D(\pi||\bu_{j-1}) - D(A\pi||A\bu_{j-1}).
\end{align}
The well-known data processing theorem \cite{Gallager} ensures that this difference is non-negative. Hence,
\begin{proposition}\label{prop:geometric_better}
$\rho^{(\GA)} \geq \rho^{(\AA)}$ almost surely.\putqed
\end{proposition}
 The only case where the limit in \eqref{eq:gamma_exact} --- the performance gap --- tends to zero is when $D(\pi||\bu_{i}) \to 0$, or simply when $\bu_i \to \pi$. We will show that if the network has sufficient connectivity, this can only happen when all agents receive the same data, i.e., $\br_{k,i} = \br_{\ell,i}$ for all agents $k$ and $\ell$. It is obvious that in this case most fusion methods, and in particular AA and GA, are equivalent. In fact, there is no need for an agent to communicate with its neighbors --- every neighbor is equivalent to the agent itself.

If the network is connected enough, and when the agents do not observe the same data, $\rho^{(\GA)}$ is \emph{strictly} greater than $\rho^{(\AA)}$, and the gap is quantified as the limit term in \eqref{eq:gamma_exact}. To study this term, we regard $\{\bu_i\}$ as a Markov chain. We denote the $K$-dimensional probability simplex as $\mathbb{S}_K$, and define the time-homogeneous transition map $T:\mathbb{S}_K \times \mathbb{R}^K  \to \mathbb{S}_K$ as 
\begin{align}
    T(u,\bR) \triangleq \frac{\bR Au}{\sum_\ell [\bR A u]_\ell}
\end{align}
where $\bR$ is the diagonal matrix with $[\bR]_{kk} = \br_{k} \triangleq \br_{k,1}$ due to time-homogeneity. Note that the $\bu_i$'s in \eqref{eq:gamma_exact} obey this map with $\bu_{i} = T(\bu_{i-1},\bR_i)$, and with $\bR_i$ independent of $\bu_{i-1}$. It can be verified that the Markov chain governed by the mapping $T$ has at least one invariant distribution $\bQ$ on $\mathbb{S}_K$, i.e., if $\bu$ has distribution $\bQ$, so has $T(\bu,\bR)$ --- this comes from the observation that $T$ is Feller continuous and from Krylov-Bogolyubov theorem for compact spaces \cite{hairer2010convergence}. However, it may be a futile attempt to find such invariant distributions. Furthermore, it is not certain if there is a unique invariant distribution although the limit in \eqref{eq:gamma_exact} exists. Hence, we resort to a different method and study a lower bound for the gap.

We now give the intuition behind our derivation for the lower bound. Since the state space of the Markov chain $\mathbb{S}_K$ is compact, there must exist a $u_0$ and a small neighborhood $\cV_0$ around it that is visited infinitely often. So, the limit in \eqref{eq:gamma_exact} is lower bounded as
\begin{align}
\lim_{i \to \infty}\frac 1 i \sum_{j=2}^i\big[D(\pi||\bu_{j-1}) - D(\pi||A\bu_{j-1})\big] &\geq\inf_{u\in \cV_0}\E[D(\pi||\bu_1) - D(\pi||A\bu_1)] \notag \\
    \label{eq:div_difference}&\geq\inf_{u\in \mathbb{S}_K}\E[D(\pi||\bu_1) - D(\pi||A\bu_1)]
\end{align}
with $\bu_1 = T(u, \bR)$. Furthermore, we use the strong data processing inequality \cite{SDPI} to obtain a further lower bound. It is then useful to give the following definition.
\begin{definition}[Contraction coefficient \cite{SDPI}] Let $A_{K\times K}$ be a probability transition matrix. Then the contraction coefficient associated with $A$ is given by
\begin{equation}\vspace{-1em}
  \eta_A \triangleq \sup_{\substack{u,v \in \mathbb{S}_K\\ u \neq v}} \frac {D(Au||Av)}{D(u||v)} \leq 1.
\end{equation}\putqed
\end{definition}
\noindent The coefficient \( \eta_A \) is lower bounded by the second largest absolute eigenvalue of \( A \) \cite{SDPI}. Hence, dense graphs usually lead to lower \( \eta_A \) that is close to 0 in value, while sparse graphs usually lead to \( \eta_A \) that is close to 1 in value. The strong data processing inequality implies that 
\begin{equation}
D(Au||Av) \leq \eta_A D(u||v)
\end{equation}
for all $K$-dimensional probability vectors $u$, $v$. Applying this inequality to \eqref{eq:div_difference}, we obtain
\begin{equation}\label{eq:infimization}
    \rho^{{(\GA)}} - \rho^{(\AA)}  \geq (1-\eta_A)\inf_{u\in \mathbb{S}_K}\E[D(\pi||\bu_1)].
\end{equation}
\textbf{Remark:} Applying the strong data processing inequality directly to \eqref{eq:gamma_exact}, it is evident that 
\begin{align}
   \lim_{i \to \infty}\frac 1 i \sum_{j=1}^i D(\pi||\bu_{j})  &\geq \rho^{{(\GA)}} - \rho^{(\AA)} \notag \\
   &\geq (1-\eta_A)\lim_{i \to \infty} \frac 1 i \sum_{j=1}^i D(\pi||\bu_{j}).
\end{align}
Therefore, when $\eta_A = 0$, we have equality. This corresponds to the case when $A$ is rank one, with every column of $A$ being $\pi$. For this special case, we will show in Sec.~\ref{sec:rankone} that $\gamma$ has a simple form. \putqed \vspace{-1em} \\ 

We now focus on the infimization problem that shows up in \eqref{eq:infimization}. Writing explicitly, it is equivalent to
\begin{align}
   \inf_{\substack{u\\v= uA^{\T}}} \!\! \E\bigg[ \!\sum_k \pi_k\log\frac{\pi_k}{v_k \br_k} + \log(\sum_k v_k \br_k)\bigg] 
   &=\!\!\!\!\inf_{\substack{u\\v= uA^{\T}}} \E\bigg[ \sum_k \pi_k\log\frac{\pi_k}{v_k} + \log(\sum_k v_k \br_k)\bigg] + \rho^{(\GA)} \notag \\\label{eq:infimization_v}
   &\geq \!\inf_{v}  \sum_k \!\pi_k\log\frac{\pi_k}{v_k} + \E\bigg[\log(\sum_k \!v_k \br_k)\bigg] \!+\! \rho^{(\GA)}
\end{align} 
We give the Karush-Kuhn-Tucker (KKT) conditions for this problem in the next result.
\begin{lemma}[Optimality conditions]\label{lem:kkt}
The KKT conditions for the optimization problem in \eqref{eq:infimization_v}, namely the infimization of
\begin{equation}\label{eq:Fv}
  F(v) \triangleq \sum_k \pi_k\log\frac{\pi_k}{v_k} + \E\bigg[\log(\sum_k v_k \br_k)\bigg] + \rho^{(\GA)},
\end{equation}
are given by
\begin{equation}\label{eq:KKT_conditions}
    \frac{\pi_k}{v_k} = \E\bigg[\frac {\br_k}{\sum_\ell \br_\ell v_\ell} \bigg] \quad \forall k.
\end{equation}
\end{lemma}
\begin{proof} See Appendix \ref{app:kkt}. 
\end{proof}
Note that $F(v)$ cannot approach the infimum over the boundary of the $K$-dimensional simplex as $F(v)$ tends to infinity close to the boundary. Using the KKT conditions above, and replacing $\dfrac {\pi_k}{v_k}$ with $\E\big[\dfrac{\br_k}{\sum_{\ell} v_\ell \br_\ell}\big]$ in the first summation term in $F(v)$, we consider the infimization of 
\begin{equation}
G(v) \triangleq \sum_k \pi_k\log\E\bigg[\frac{\br_k}{\sum_{\ell} v_\ell \br_\ell}\bigg] + \E\bigg[\log(\sum_k v_k \br_k)\bigg] + \rho^{(\GA)}\label{eq:Gv}.
\end{equation}
Observe that $\inf_v F(v) \geq \inf_v G(v)$. We summarize the above results in the theorem below.
\begin{theorem}[Variational lower bound to the gap]\label{thm:gap}The performance gap is lower bounded as
\begin{align}
\rho^{(\GA)}-\rho^{(\AA)}  \geq (1-\eta_A) \inf_v F(v) \geq (1-\eta_A) \inf_v G(v)
\end{align}
with $F(v)$ defined in \eqref{eq:Fv} and $G(v)$ in \eqref{eq:Gv}. Furthermore if $\eta_A < 1$, then
$\rho^{(\GA)} = \rho^{(\AA)}$ if and only if $\br_k\!=\!\br_\ell$ for all $k,\ell$.
\end{theorem}
\begin{proof}
See Appendix \ref{app:jensen}. 
\end{proof}
\noindent\textbf{Remark:} The bound we have provided in this section partially captures the effect of network structure through the quantity $\eta_A$. It is known that $\eta_A \leq \dob_A$, the Dobrushin coefficient of $A$ \cite{SDPI}:
\begin{equation}\label{eq:dobrushin_coeff}
    \dob_A \triangleq \max_{k,\ell}\frac 1 2\|a_k-a_\ell\|_1
\end{equation}
where \( a_k \) denotes the \(k\)th column of the matrix \(A\). Furthermore, $\eta_A < 1$ if and only if $\dob_A < 1$ \cite{SDPI,AG76}. Hence, for the broad class of combination matrices where $\dob_A < 1$, i.e., geometrically ergodic ones \cite[Chapter 2.7]{krishnamurthy_2016}, the bound is non-trivial. However when $\dob_A = 1$, Theorem \ref{thm:gap} yields the trivial bound $\rho^{(\GA)} \geq \rho^{(\AA)}$. This implies that whenever there are two agents $k$, $\ell$ with non-overlapping sets of neighbors, i.e., $\cN_k\cap \cN_\ell = \emptyset$, we obtain the trivial bound. One might attempt to strengthen the bound by replacing $\eta_A$ with
\begin{equation}
\tilde\eta_A \triangleq \sup_{v\neq \pi} \frac{D(\pi||Av)}{D(\pi||v)}\leq\eta_A.
\end{equation}
However, to the best of our knowledge, there may not be a straightforward method to calculate this quantity.\putqed \vspace{-1em}\\

Theorem \ref{thm:gap} points out that the decay rate of AA-diffusion is highly dependent on network connectivity via $\eta_A$ as opposed to the decay rate of GA-diffusion, whose decay rate only depends on the network centrality, i.e., Perron vector $\pi$. In \cite{jadbabaie2013information}, the authors study the effects of network regularity under AA-consensus on the upper bound $\rho^{(\GA)}$. Different from their work, the bound on the performance gap in Theorem \ref{thm:gap} captures the effect of network connectivity.
\subsection{Some Special Cases}\label{sec:special_cases_exact}
The bounds given in the previous sections are in variational form. Although we have studied certain characteristics of these bounds, i.e., found the KKT conditions, still, the bounds are highly dependent on the joint distribution of the data across the users --- recall that the decay rate of GA-diffusion, $\rho^{(\GA)}$, has a closed form expression and only depends on the marginals. Moreover, the extremal points satisfying the KKT conditions are difficult to find in general. Hence, we study two special cases in this section.
\subsubsection{Rank-one Combination Matrices}\label{sec:rankone}
When $A$ is a rank one matrix, it turns out that $\rho^{(\AA)}$ has a closed form. In this case, $A$ can be written as $A = \pi \mathds{1}_K^{\T}$. Then,
\begin{align}\label{eq:rank_one}
    A^{\T}\bR A^{\T} &= \mathds{1}_K\pi^{\T}\bR \mathds{1}_K\pi^{\T}  \notag \\
    &=\big(\sum_k \pi_k \br_k\big)\mathds{1}_K \pi^{\T} \notag \\&=\big(\sum_k \pi_k \br_k\big)A^{\T}
\end{align}
and it follows that
\begin{equation}
    \bY_i = \prod_{j = 1}^{i-1} (A^{\T}\bR_j)(A^{\T}\bR_i) = \prod_{j = 1}^{i-1}\big(\sum_k \pi_k \br_{k,j}\big) (A^{\T}\bR_i).
\end{equation}
If we choose $k$, $\ell$ such that $a_{\ell k} > 0$, it holds that
\begin{align}
    \frac 1 i \log [\bY_i]_{k\ell} = \frac 1 i \sum_{j = 1}^{i-1} \log\big(\sum_k \pi_k \br_{k,j}\big) + \frac 1 i \log({a_{\ell k}\br_{\ell,i}})
\end{align}
and from the law of large numbers
\begin{align}
    \lim_{i \to \infty} \frac 1 i \log [\bY_i]_{k\ell} = \E\Big[\log\big(\sum_k \pi_k \br_k\big)\Big].
\end{align}
Corollary \ref{cor:existence_gamma} and Theorem \ref{thm:existence} then lead to the following.
\begin{proposition}[Exact rate under rank-one topologies]\label{prop:rankone}
\begin{equation}
   \rho^{(\AA)} = -\gamma = -\E\Big[\log\big(\sum_k \pi_k \br_k\big)\Big].
\end{equation}\putqed
\end{proposition}
\noindent Notice from above that the performance gap
\begin{equation}\label{eq:jensen_gap}
    \rho^{(\GA)} - \rho^{(\AA)} = \E\Big[\log\big(\sum_k \pi_k \br_k\big)-\sum_k \pi_k \log\br_k\Big]
\end{equation}
is equal to the expectation of a Jensen's inequality gap. \vspace{-1em}\\

Fully-connected networks with rank-one combination matrices are equivalent to architectures with fusion center in terms of performance, because each agent computes the same weighted average of all beliefs across the network. In this sense, each agent acts like a fusion center. Therefore, this special case result is of interest for inference problems with central processors. Consider a federated system where ($i$) agents send their beliefs to the center, ($ii$) the center averages the peripheral agents' beliefs with importance sampling, and ($iii$) sends the aggregated belief back to the agents. The learning rate of such system is given in Proposition \ref{prop:rankone}.

\subsubsection{Exchangeable Networks}
In this special case, we assume that the data is exchangeable across the users. More precisely, $\{\bxi_{1}(\theta),\dots,\bxi_{K}(\theta)\}$ constitutes a set of exchangeable random variables for every $\theta \in \Theta$.
\begin{definition}[Exchangeable random variables] A set of random variables is called exchangeable if the distribution of
$\bxi_{1},\dots,\bxi_{K}$ remains unchanged under any permutation over the index set, i.e., for all $\xi_1,\dots, \xi_k$
\begin{equation}
    \bP{\bxi_{1}\leq \xi_1,\dots,\bxi_{K}\leq \xi_k}  = \bP{\bxi_{\sigma_1}\leq \xi_1,\dots,\bxi_{\sigma_K}\leq \xi_k}
\end{equation}
with $\sigma$ being any permutation of $\{1,\dots,K\}$. \putqed
\end{definition}

Exchangeability is a weaker assumption than i.i.d., as the data need not be independent across the agents. Observe that exchangeable networks could be of particular interest as they model fair networks --- since exchangeability requires identical distributions of data across the agents, no agent learns better than another if there was no cooperation.
For this particular example, we also assume that $A$ is doubly stochastic, i.e., $\sum_{k} a_{\ell k} = 1$ as well. Then, it is easily seen that each element of the Perron vector of $A$ is $\pi_k = 1/K$. Under this assumption, we are able to solve the KKT conditions of the problem \eqref{eq:infimization_v}.
\begin{theorem}[Lower bound under exchangeable networks]\label{thm:exchangeable}If the data across the agents is exchangeable and $A$ is doubly stochastic,
$\pi$ is the unique solution of \eqref{eq:infimization_v}. Hence,
\begin{align}
    \rho^{(\AA)} \leq \eta_A\rho^{(\GA)} - (1-\eta_A)\E\Big[\log\Big(\frac 1 K\sum_k \br_k\Big)\Big]
\end{align}
and moreover,
\begin{align}
     \rho^{(\GA)}-\rho^{(\AA)} &\geq (1-\eta_A) \left(\rho^{(\GA)} +\E\Big[\log\Big(\frac 1 K\sum_k \br_k\Big)\Big] \right) \notag \\
     &=(1-\eta_A) \E\Big[\log\big(\sum_k \frac{1}{K} \br_k\big)\!-\!\!\sum_k \frac{1}{K} \log\br_k\Big]
\end{align}
\end{theorem}
\begin{proof}
It immediately follows that $\br_1,\dots,\br_K$ is also exchangeable. Also, the KKT conditions in \eqref{eq:KKT_conditions} imply
\begin{equation}
\E\Bigg[\frac{\br_1v_1}{\sum_\ell \br_\ell v_\ell}\Bigg] = \E\Bigg[\frac{\br_2v_2}{\sum_\ell \br_\ell v_\ell}\Bigg]
\end{equation}
and because of exchangeability,
\begin{align}\label{eq:KKT_difference}
\E\Bigg[\frac{\br_1v_1}{\br_1v_1 + \br_2v_2 + \bss}\Bigg] = \E\Bigg[\frac{\br_1v_2}{\br_2v_1 + \br_1v_2 + \bss}\Bigg]
\end{align}
where $\bss \triangleq \sum_{\ell > 2}\br_\ell v_\ell$. Suppose that $v_1 > v_2$. The difference of the two terms in \eqref{eq:KKT_difference} then becomes \begin{align}
&\E\Bigg[\frac{\br_1v_1}{\br_1v_1 + \br_2v_2 + \bss} - \frac{\br_1v_2}{\br_2v_1 + \br_1v_2 + \bss}\Bigg] \notag \\
&=  \!\!(v_1-v_2)\E\Bigg[\frac{\br_1\br_2(v_1+v_2) + \br_1\bss}{(\br_1v_1 + \br_2v_2 + \bss)(\br_2v_1 + \br_1v_2 + \bss)}\Bigg].
\end{align}
Since we assumed $v_1 > v_2$, the difference in \eqref{eq:KKT_difference} must be strictly greater than zero. A similar argument for $v_2 > v_1$ yields the same result. This implies $v_1 = v_2$, and by symmetry, $v_k = v_\ell$ for all $k,\ell$. Therefore, $v$ must be equal to $\pi^{\T}$.
\end{proof} 

It is seen from Theorem \ref{thm:exchangeable} that the lower bound on the performance gap increases as the network becomes more connected, i.e., small $\eta_A$. Moreover, the Jensen's inequality gap in \eqref{eq:jensen_gap} and Theorem \ref{thm:exchangeable} increases when the observations become more diverse across the agents. The performance gap drifts away from zero. This results in an interesting observation: If the individual agents, or any subgroup of $M$ agents, have the same learning abilities (implied by exchangeability), the overall learning speed of the network decreases with respect to the increased amount of collaboration under AA-diffusion. However, the learning abilities under GA-diffusion is not affected --- every agent in the network learns at the same speed regardless of the amount of collaboration.

\section{\textbf{Numerical Results}}\label{sec:numerical}

In this section, we provide numerical results to study the gap between AA and GA-diffusion; and to study the effect of network connectivity on the decay rate of AA-diffusion. We simulated the networks given in Figure \ref{fig:network2reg}, \ref{fig:network4reg}, \ref{fig:networknotreg}. All networks consist of $K=10$ nodes. The network \ref{fig:network2reg} is 2-regular, \ref{fig:network4reg} is 3-regular and \ref{fig:networknotreg} is not a regular graph. Recall that a graph is called $D$-regular when each vertex has $D$ neighbors.

\begin{figure*}[hbt!]
	\centering
	\subfloat[a][]{
  \includegraphics[width=.25\linewidth]{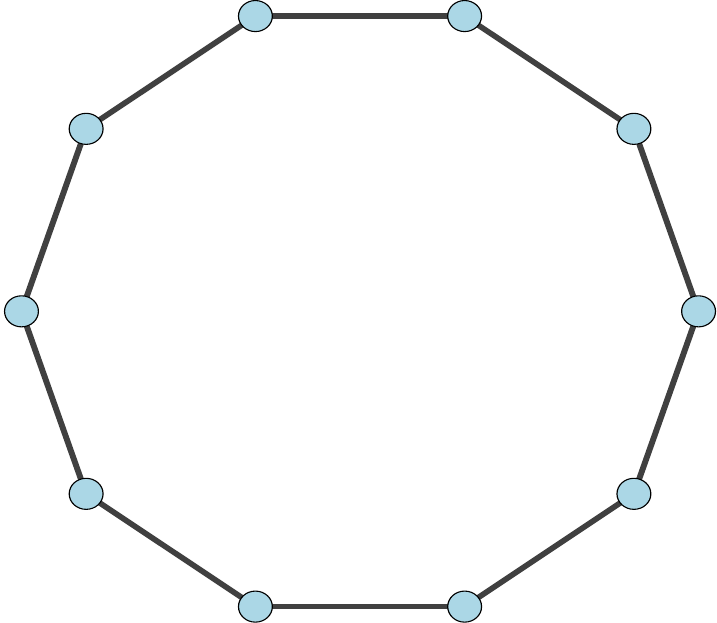}\label{fig:network2reg}
  }\hfil
  \subfloat[b][]{
  \includegraphics[width=.25\linewidth]{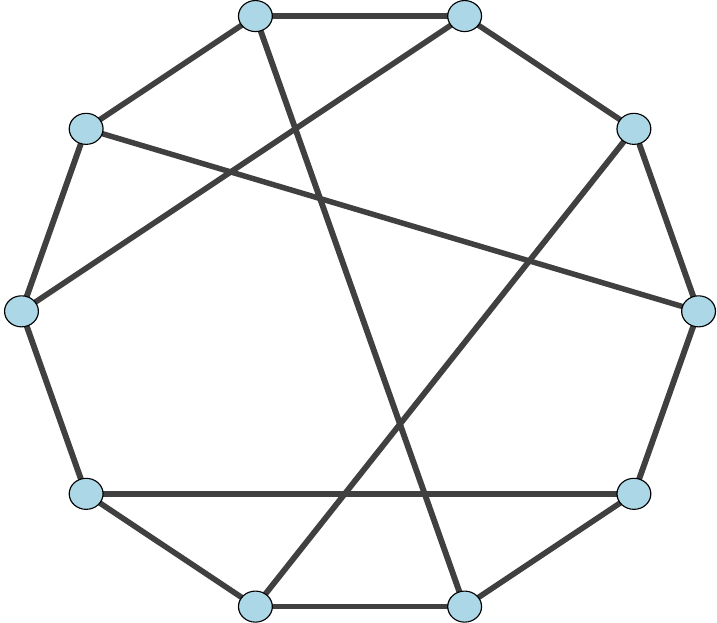}\label{fig:network4reg}
  }\hfil
  \subfloat[c][]{
  \includegraphics[width=.25\linewidth]{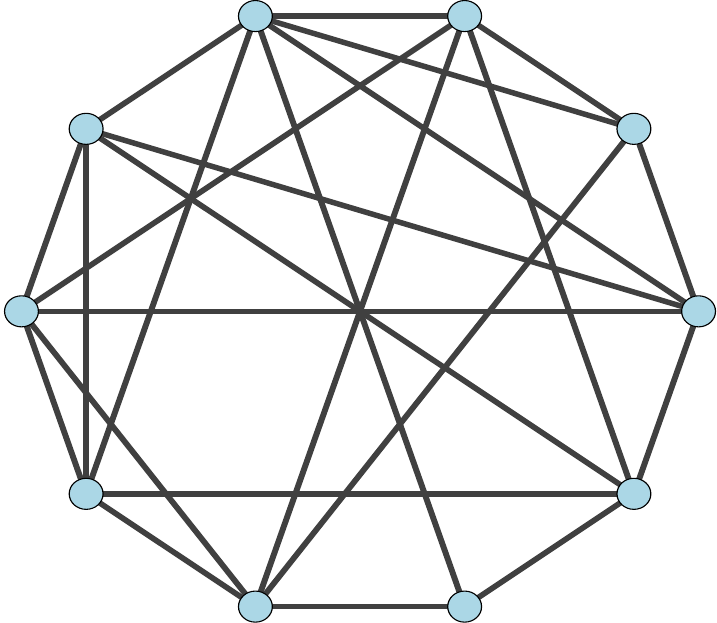}\label{fig:networknotreg}
  }
  \caption{\small The three networks consist of $K=10$ nodes. ($a$) is a 2-regular network, ($b$) is a 3-regular network and ($c$) is non-regular with 24 edges. Self-loops are omitted for visual simplicity in all figures.}
 \end{figure*}
 
Our first simulation compares the 2-regular and 3-regular networks. The combination matrices are denoted by $A_2$ and $A_3$, and are set as follows. For an $\alpha \in [0,1]$, if nodes $\ell$ and $k$ are connected, then $[A_2]_{\ell k} = \frac {1-\alpha} 2$ and $[A_3]_{\ell k} = \frac {1-\alpha} 3$ respectively; and $[A_2]_{\ell \ell} = [A_3]_{\ell \ell} = \alpha$. The other elements are necessarily set to zero. Observe that we ensure the diagonal elements of $A_2$ and $A_3$ are the same and equal to $\alpha$ --- this setting enables us to compare the decay rates with the bound given in \cite{jadbabaie2013information}. We have chosen $\alpha = 0.05$. Note that $A_2$ and $A_3$ are doubly stochastic and thus their Perron vectors $\pi$ have entries $\pi_k = 1/K$.
Moreover, we assume $H =2$. Under the true hypothesis $\theta^\circ$, the agents observe exponential random variables with parameter 1, and under the alternative hypothesis $\theta$, each agent $k$ observes exponential random variables with parameter $\beta_k$, with $\beta_k \in [0.500, 0.300, 0.025, 0.750, 1.200, 2.250, 0.900, 1, 0.250, 0.025]$. For simplicity, we assume that the data is independent across the agents.

The results of the first experiment are given in Figure \ref{fig:reg_comparison}. We have plotted agent 1's belief decay rates on hypothesis $\theta$, with AA and GA-diffusion algorithms (denoted with the superscript AA and GA, respectively) and on networks \ref{fig:network2reg} and \ref{fig:network4reg} (denoted with the subscript added to the algorithm description). For instance, the $\AA_3$ in the superscript refers to the decay rate with AA-diffusion and on the 3-regular network \ref{fig:network4reg}. First, observe the significant performance gap between GA and AA-diffusion learning rates. As expected, the decay rate of GA-diffusion is not affected by the network regularity --- we know it only depends on the Perron vector $\pi$. However, the AA-diffusion decay rate is visibly affected, which is expected according to Theorem \ref{thm:gap}. As a final remark, we point out that in \cite{jadbabaie2013information}, the authors upper bounded the decay rate of AA with consensus algorithm (not AA-diffusion) with $\alpha \rho^{(\GA)}$. This is not true for AA-diffusion as for our setting $\alpha = 0.05$, $\rho^{(\GA)} = 0.7261$, and $\rho^{(\AA_2)}$, $\rho^{(\AA_4)}$ both seem to be above $\alpha \rho^{(\GA)} = 0.0363$. This also shows that since the AA-diffusion decay rate is above the upper bound given for AA-consensus, the agents learn faster with the diffusion algorithm than with consensus. This behavior is consistent with the results in \cite{tu2012}, which showed that diffusion strategies are superior to consensus strategies in terms of performance in distributed estimation. Figure \ref{fig:reg_comparison} complements this result in the sense that it shows diffusion outperforms consensus in the AA-social learning setting as well. Such distinction was not present for GA-social learning at least asymptotically. The GA-diffusion and GA-consensus have the same asymptotic learning rate.

 \renewcommand{\thesubfigure}{\roman{subfigure}}
\begin{figure*}[hbt!]
	\centering
	\subfloat[a][]{
  \includegraphics[width=.4\linewidth]{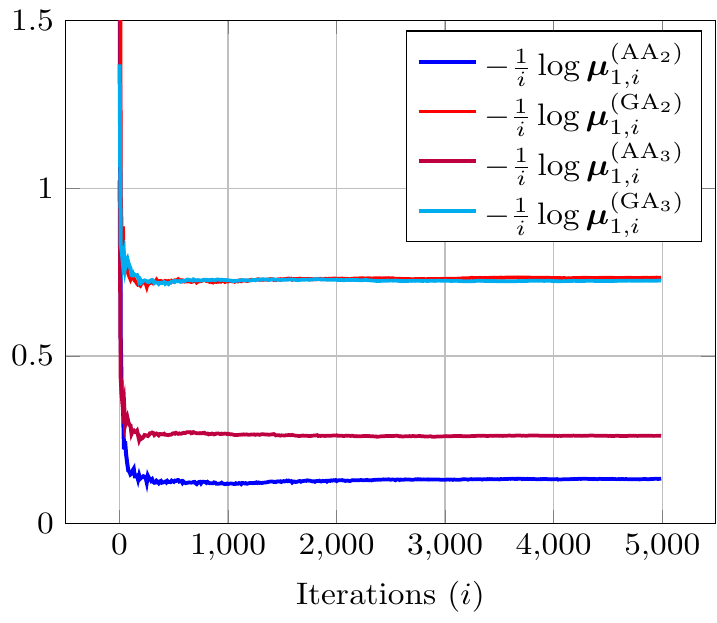}\label{fig:reg_comparison}
  }\hfil
  \subfloat[b][]{
  \includegraphics[width=.4\linewidth]{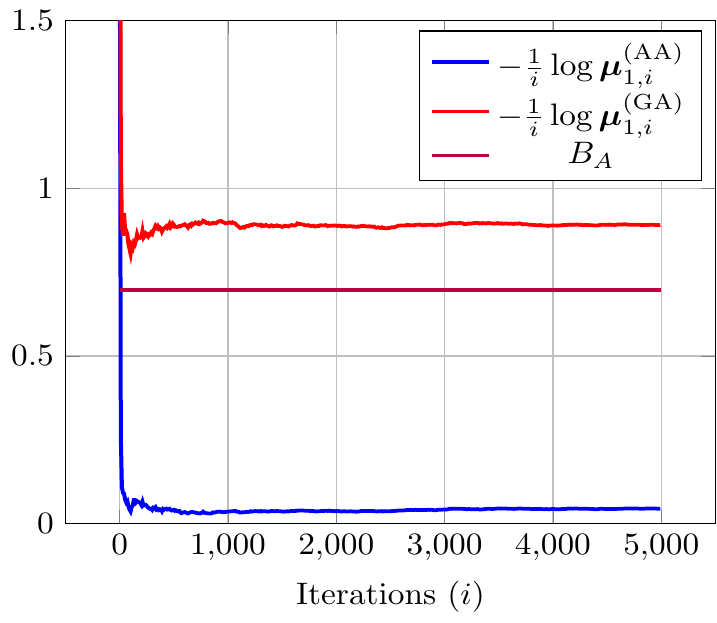}\label{fig:conjecture}
  }\hfil
  \subfloat[c][]{
  \includegraphics[width=.4\linewidth]{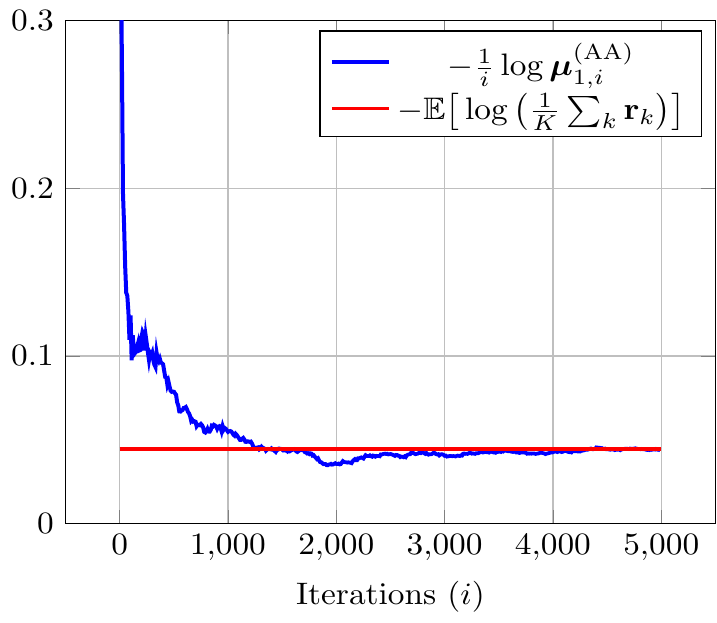}\label{fig:rankone}
  }\caption{\small (i) Comparison of the decay rates of different network connectivities. The superscripts on the decay rates indicate the algorithm, and the network where we execute the algorithm, e.g., $\AA_2$ means AA-diffusion is executed on the 2-regular network \ref{fig:network2reg} with combination matrix $A_2$. The decay rates of GA-diffusion are the same for networks \ref{fig:network2reg} and \ref{fig:network4reg} whereas they differ when AA-diffusion is run. (ii) The decay rates of the non-regular network \ref{fig:networknotreg} are plotted under AA and GA-diffusion in the exchangeable setting; and the closed form bound $B_A$ is shown. (iii) Time-scaled minus log-belief of agent 1 in the fully connected network with a rank-one combination matrix. This entity approaches the closed form expression we provided in Proposition \ref{prop:rankone}.}
\end{figure*}

The second experiment involves an exchangeable network. To this end, we set $\beta_k = 3$ for all $k$. Referring to the result in Theorem \ref{thm:exchangeable}, we define
\begin{equation}
    B_A \triangleq \dob_A \rho^{(\GA)} - (1-\dob_A)\E\Big[\log\Big(\frac 1 K\sum_k \br_k\Big)\Big].
\end{equation}
Recall that $\dob_A$ was defined in \eqref{eq:dobrushin_coeff}. We simulated AA-diffusion on the non-regular network $\ref{fig:networknotreg}$ with its combination matrix chosen according to a lazy Metropolis rule \cite{Metropolis1953}. More precisely, we take $B \triangleq [b_{\ell k}]$ with $b_{\ell k} = \max\{\deg(\ell),\deg(k)\}^{-1}$ for $\ell \neq k$ and $b_{\ell \ell} = 1-\sum_{k\neq\ell}b_{\ell k}$. Then, we set $A_{\text{non}} =  \alpha I + (1-\alpha) B$. The matrix $A_{\text{non}}$ is also doubly stochastic, hence, its Perron vector is the same as networks \ref{fig:network2reg} and \ref{fig:network4reg}. We have plotted the decay rates of agent 1 with the combination matrix $A_{\text{non}}$ and also indicated the bound $B_A$ in Figure \ref{fig:conjecture}. Since the Dobrushin coefficient $\dob_A = 0.81$, $B_A$ gives a non-trivial bound.

The third experiment includes a fully-connected network with combination matrix $a_{\ell k} = 1/K$. Observe $A$ is rank-one and from Proposition \ref{prop:rankone}, we know
\begin{equation}
    \rho^{(\AA)} = -\E\Big[\log\Big(\frac 1 K\sum_k \br_k\Big)\Big] \approx 0.0457.
\end{equation}
The AA-diffusion decay rate corresponding to the final experiment is plotted in Figure \ref{fig:rankone} and it is readily seen that the minus log-belief with time scaling approaches 0.0457.

The next part of this section consists of numerical examples where the ``inept agent'' phenomenon is observed --- recall the end of Section \ref{sec:subadditive}. We set the inept agent to agent 1, and select $\beta_1 = 1$, i.e., $D(L_1(\cdot|\theta^\circ)||L_1(\cdot|\theta)) = 0$. The other $\beta_k$'s remain the same as in the simulations in Figure \ref{fig:reg_comparison}. The simulations take place over the non-regular network \ref{fig:networknotreg}, and four different connecivity matrices $A_{\alpha_1}$, $A_{\alpha_2}$, $A_{\alpha_3}$, $A_{\alpha_4}$ are set with the same lazy Metropolis rule in \ref{fig:conjecture}, with four different $\alpha$ values $\alpha_1 = 0.01$, $\alpha_1 = 0.5$, $\alpha_1 = 0.8$, $\alpha_4 = 0.95$. Note that the Perron vector remains unchanged when we change $\alpha$, hence the decay rate of GA should also remain unchanged. This is observed in Figure \ref{fig:inept}. However, as $\alpha$ increases, the inept agent 1 becomes more self-confident and the learning rate for AA decreases drastically, which is evident from Figure \ref{fig:inept}.

\begin{figure*}[hbt!]
	\centering
	\subfloat[a][]{
  \includegraphics[width=.35\linewidth]{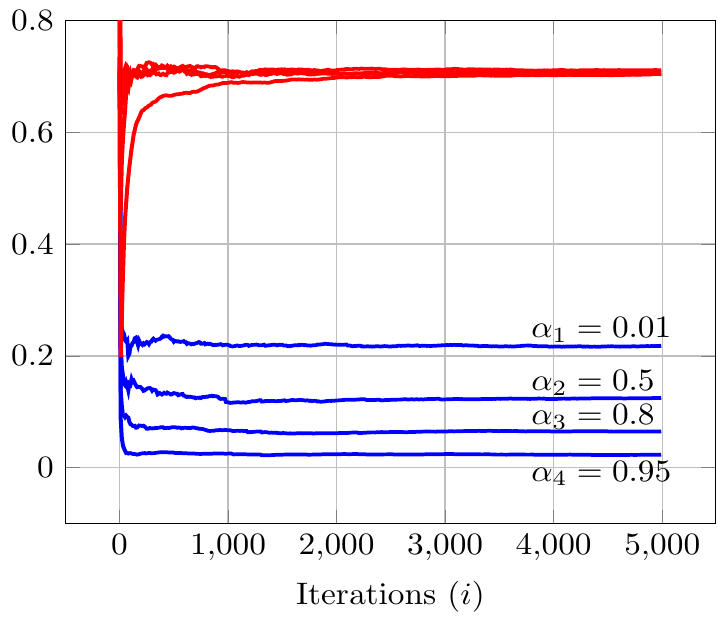}\label{fig:inept}
  }\hfil
  \subfloat[b][]{
  \includegraphics[width=.35\linewidth]{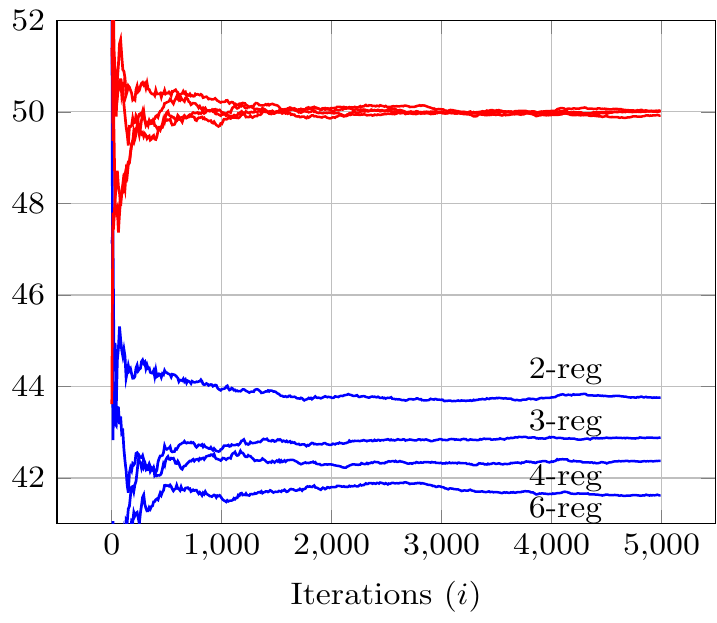}\label{fig:exch_regular}
  }\hfil
  \subfloat[c][]{
  \includegraphics[width=.35\linewidth]{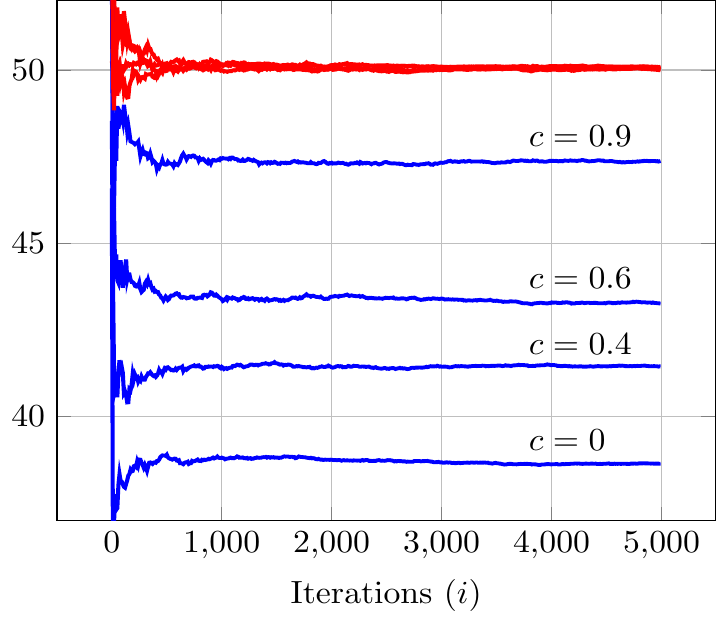}\label{fig:exch_corr}
  }
  \caption{\small ($i$) Numerical example for the ``inept agent'' phenomenon mentioned in the end of Section \ref{sec:subadditive}. We have drawn the time-scaled minus log-beliefs, i.e., $-\frac 1 i\log \bmu_{1,i}$'s of agent 1 over the network \ref{fig:networknotreg} and with connectivity matrices $A_{\alpha_1}$, $A_{\alpha_2}$, $A_{\alpha_3}$, $A_{\alpha_4}$. The red curves correspond to the decay rates of GA-diffusion and the blue curves correspond to the decay rates of AA-diffusion. The decrease in the AA-diffusion decay rate with respect to the increase in $\alpha$ is evident. ($ii$) Comparison of different network connectivities in an exchangeable setting. The red curves correspond to the decay rates of GA-diffusion and the blue curves correspond to the decay rates of AA-diffusion simulated on 2,3,4,6-regular networks. It is clearly visible that the decay rate decreases with network connectivity. ($iii$) Comparison of decay rates under different joint distributions for the Gaussian setting. The red curves correspond to the decay rates of GA-diffusion and the blue curves correspond to the decay rates of AA-diffusion simulated for $\Sigma(c)$, $c=0,0.4,0.6,0.9$.}
 \end{figure*}

For the remaining simulations, we assume that, under $\theta^\circ$, agents observe standard Gaussian random variables and under $\theta$ they observe unit-variance Gaussian random variables with mean 10. The $K \times K$ covariance matrix of the data under both hypotheses is given by 
\begin{equation}\label{eq:exp_sigmac_def}
    \Sigma(c)\triangleq c\mathds{1}_K\mathds{1}^{\T}_K + (1-c)I_{K}
\end{equation}
for a $c \in [0,1]$. Observe that $c = 1$ corresponds to the case where all agents observe the same data and $c = 0$ corresponds to the i.i.d. case. This setup ensures that the data is exchangeable.

The following simulation aims to investigate the effect of network connectivity on the decay rates of AA-diffusion under exchangeable networks. We choose $c = 0.5$, and the network is simulated on $D = $ 2,3,4,6-regular networks. The $D$ regular networks are constructed as follows: The neighbor of agent $k$ is set as 
\begin{equation}
    \cN_k = \{k-\lfloor D/2\rfloor,\dots k-1,k+1,\dots,k+\lceil D/2\rceil\} \mod{K}.
\end{equation}
Note that the 2-regular network constructed with this method is equivalent to that in Figure \ref{fig:network2reg} whereas the 3-regular network is different from the one in Figure \ref{fig:network4reg}. The connectivity matrices are set as in \ref{fig:reg_comparison} with the same $\alpha = 0.05$. In line with Theorem \ref{thm:exchangeable}, Figure \ref{fig:exch_regular} illustrates that the performance gap increases with connectivity. Finally,  we compare the decay rates under various joint distributions under exchangeable setting. We simulate AA-diffusion under the same Gaussian setting for $c = 0,0.4,0.6,0.9$. Recall that the covariance matrix was given by $\Sigma(c)$ from \eqref{eq:exp_sigmac_def}. We also observe that the decay rate increases when $c$ increases, and in particular, is equal to the decay rate under GA-diffusion for $c = 1$ --- every agent observes same data if $c=1$.

\section{\textbf{Discussion}}
In this work, we compared the decay rates under arithmetic and geometric averaging of beliefs for social learning over networks. For arithmetic averaging strategies, we established that the beliefs on wrong hypotheses decay exponentially almost surely, and the decay rates are the same among all agents. We provided upper and lower bounds for the decay rates, one being the decay rate corresponding to geometric averaging (Proposition \ref{prop:geometric_better}) and the other one revealing an interesting phenomenon (at the end of Section \ref{sec:subadditive}), which we called the ``inept agent'' phenomenon. Specifically, if there is a highly self-confident agent in the network whose learning abilities are limited, the decay rates of every single agent in the network drop significantly. Such phenomenon is particular to arithmetic averaging and can be mitigated with geometric averaging if other agents follow a skeptical approach towards the information conveyed by the inept agent.

We studied the performance gap between arithmetic and geometric averaging in detail with an appropriate formulation that permitted the use of the strong data processing inequality. In turn, we have proved that for a broad class of networks, there is no performance gap if, and only if, all agents observe the same data. Furthermore, we have obtained closed form bounds and expressions for the performance gap for some special instances. The strong data processing coefficient $\eta_A$ may not be sufficiently tight for our analyses and it might be replaced with $\tilde \eta_A$ for obtaining tighter bounds for our problem. The calculation of $\tilde \eta_A$, however, seems challenging and might be an interesting problem on its own.

An interesting future direction is to examine how our results relate to the distributed estimation and filtering. Such methods primarily focus on the inference of continuous variables, and they typically rely on the fusion of point estimates. Nevertheless, they can be interpreted in terms of fusion of probability density functions \cite{koliander2022fusion}. Therefore, a natural question is the extend of applicability of our results for social learning to that line of work. Furthermore, all results in this work are asymptotic. In the subsequent work \cite{ssp_federated_2023}, we established the asymptotic normality of AA for federated architectures. Employing the Berry-Esseen theorem \cite{berry1941accuracy} and related tools on top of this work to understand finite sample behavior of the AA algorithm can be another interesting direction to study.

\appendix

\subsection{Behavior of $\brho_k$'s characterize the behavior of $\brho$}\label{app:A}
We first give a different form for \eqref{eq:TV_regret}:
\begin{equation}
  \liminf_{i \to \infty}\frac 1 i \bigg|\log \bigg(\sum_{k = 1}^K\mathbf{d}_{k,i}\bigg)\bigg| =-\limsup_{i \to \infty}\frac 1 i \bigg(\frac 1 K\log \bigg(\sum_{k = 1}^K\mathbf{d}_{k,i}\bigg)\bigg)
\end{equation}
as 
\begin{equation}
    \mathbf{d}_{k,i} = \frac 1 2\|\bmu_{k,i}(\theta) - e_{\theta^\circ}\|_1 \leq 1.
\end{equation}
The term above is equivalent to
\begin{align}
\limsup_{i \to \infty}\frac 1 i \log \bigg(\frac 1 K \sum_{k=1}^K \sum_{\theta\neq \theta^\circ} \bmu_{k,i}(\theta)\bigg).
\end{align}
Moreover, observe that
\begin{align}
 -\frac 1 i\log K + \max_{\substack{1\leq k\leq K \\ \theta \neq \theta^\circ}}  \frac 1 i \log \bmu_{k,i}(\theta) 
 &\leq \frac 1 i \log \bigg(\frac 1 K \sum_{k=1}^K \sum_{\theta\neq \theta^\circ} \bmu_{k,i}(\theta)\bigg) \notag \\
 &\leq \frac 1 i \log H + \frac 1 i \max_{\substack{1\leq k\leq K \\ \theta \neq \theta^\circ}} \frac 1 i \log \bmu_{k,i}(\theta),
\end{align}
and since $\dfrac 1 i \log H$ and $\dfrac 1 i \log K$ vanish as $i$ grows, the quantities of interest are
\begin{equation}
    \brho_{k}(\theta) = \limsup_{i \to \infty} \frac 1 i \log \bmu_{k,i}(\theta)
\end{equation}
and 
\begin{equation}
    \brho = \max_{\substack{1\leq k\leq K \\ \theta \neq \theta^\circ}}  \brho_{k}(\theta).
\end{equation}

\subsection{Proof of Corollary \ref{cor:existence_gamma}}\label{app:existence_gamma}

To make use of Theorem \ref{thm:Kingman}, we first need to ensure that the random matrices have all positive entries. Recall that there must exist an $n \geq 1$ such that every entry of $A^n$ is strictly positive. We choose the smallest such $n$. Assumption 1 ensures that the $\{\br_{k,i}\}$ are strictly positive with probability 1 as well --- otherwise some of the KL divergences would be infinite. We therefore replace $\bX_i$'s with the expression
\begin{equation}
    \widetilde\bX_i \triangleq A^{\T}\bR_{n(i-1)+1}\dots A^{\T}\bR_{ni},
\end{equation}
which turn out to have all positive entries. To see this, observe that 
\begin{align}
    [\widetilde \bX_i]_{11} = \sum_{\ell_1,\dots,\ell_{n-1}} a_{\ell_1 1 }a_{\ell_2\ell_1}\dots a_{1\ell_{n-1}}\br_{\ell_1,1}\dots\br_{1,n}.
\end{align}
Since $A$ is primitive, there must exist $\ell_1,\dots ,\ell_{n-1}$ such that
$a_{\ell_1 1 }\dots a_{1\ell_{n-1}} > 0$. If we choose such a path, we can observe that
\begin{align}\label{eq:primitive}
    [\widetilde \bX_i]_{11} \geq  a_{\ell_1 1 }a_{\ell_2\ell_1}\dots a_{1\ell_{n-1}}\br_{\ell_1,1}\dots\br_{1,n}.
\end{align}
As mentioned, $\br_{k,i}$'s are strictly positive and thus $[\widetilde \bX_i]_{11}$ is strictly positive as well. Similar arguments hold for other $[\widetilde \bX_i]_{\ell k}$'s as well, which proves that $\widetilde \bX_i$ is all positive. Note that $\widetilde \bX_i$'s are also i.i.d.. The next step is to check if the logarithms of the entries of $\widetilde\bX_i$ have finite expectations. First of all, since $\E[\bR_i] = I$ and $\bR_i$'s are i.i.d., it holds that
\begin{align}
    \E[\widetilde\bX_i] = \E[A^{\T}\bR_{n(i-1)+1}\dots A^{\T}\bR_{ni}] = (A^{\T})^n,
\end{align}
which further implies that 
\begin{align}
    \E[\log [\widetilde \bX]_{\ell k}] \leq \log\E[[\widetilde \bX]_{\ell k}] = \log[A^n]_{k\ell} \leq 0.
\end{align}
Therefore, the expectations of the logarithms are bounded from above. To check if they are also bounded from below, if we take the logarithms of both sides in \eqref{eq:primitive}, we obtain
\begin{equation}
    \E[\log[\widetilde \bX_i]_{11}] \geq  \E\log(a_{\ell_1 1}a_{\ell_2\ell_1}\dots a_{1\ell_{n-1}}) + \E[\log\br_{\ell_1,1}] + \dots + \E[\log\br_{1,n}].
\end{equation}
By Assumption 1, it is true that
\begin{equation}
\E[\log\br_{k,1}] = -D(L_k(.|\theta^{\circ}) || L_k(.|\theta))> -\infty.
\end{equation}
Consequently, $\E[\log[\widetilde \bX_i]_{11}]$ is bounded from below. A similar argument works for other $\E[\log[\widetilde \bX_i]_{\ell k}]$'s as well. 

Since $[\widetilde \bX_i]_{\ell k}$'s are strictly positive and their logarithms have finite expectations, we can invoke Theorem \ref{thm:Kingman} to obtain
\begin{equation}
\bm{\tilde\gamma} =  \lim _{i \rightarrow \infty} \frac 1 i \log [\widetilde\bY_{i}]_{\ell k}
\end{equation}
with $\widetilde \bY_i \triangleq \prod_{j = 1}^i \widetilde \bX_j$. Since $\widetilde \bY_i = \bY_{ni}$, the above equation implies
\begin{equation}
     \bm{\gamma} = \frac {\bm{\tilde\gamma}}{n} = \lim _{i \rightarrow \infty} \frac 1 {ni} \log[\bY_{ni}]_{\ell k} = \lim _{j \rightarrow \infty} \frac 1 j \log [\bY_{j}]_{\ell k}.
\end{equation}
Moreover, since $\{A^{\T}\bR_{i}\}$ is an i.i.d. sequence, Kolmogorov's zero-one law \cite{williams_1991} implies that the finite limit $\bm{\gamma}$ is almost surely a constant. Hence, 
\begin{equation}
\gamma \triangleq \E[\bm{\gamma}] = \lim _{i \rightarrow \infty} \frac 1 i \E\big[\log [\bY_{i}]_{\ell k}\big].
\end{equation}
and the proof is complete.

\subsection{Proof of Theorem \ref{thm:existence}}\label{app:existence} It is sufficient to establish the result for one agent $k$; a similar argument applies to the other agents. We establish the proof in two parts:
\begin{itemize}
    \item $(i)$ $\limsup\limits_{i \to \infty} \dfrac 1 i \log \bmu_{k,i} \leq \gamma$,
    \item $(ii)$ $\liminf\limits_{i \to \infty} \dfrac 1 i \log \bmu_{k,i} \geq \gamma$.
\end{itemize}
The part $(i)$ of the proof makes use of the extremal process $\{\bnu_{i}\}$. Setting a $\delta > 0$, we define the events
\begin{equation}
    \cH_{i_0}^{+}(\delta) \triangleq \Big\{\omega \in \Omega: \exists i_1\geq i_0, \forall i \geq i_1,
    \frac 1 {i-i_0} \max_{\ell,k}\log [\bY_{i_0}^i]_{\ell k} \leq \gamma+\delta\Big\}
\end{equation}
for every $i_0\geq 1$ with 
\begin{align}\label{eq:matrix_product}
\bY_{i_0}^i \triangleq  \prod_{j = i_0+1}^i (A^{\T}\bR_j).
\end{align}
In words, $\cH_{i_0}^{+}(\delta)$ is the event that the logarithms of all entries of ${\bY_{i_0}^i}$ eventually become smaller than $(\gamma+\delta)(i-i_0)$. Since $\{A^{\T}\bR_i\}$ is i.i.d., $\bY_i$ is stationary; and $\bP{\cH_{i_0}^{+}(\delta)}$ does not depend on $i_0$. Corollary \ref{cor:existence_gamma} states that $\bP{\cH_{0}^{+}(\delta)} = 1$, and we deduce that $\bP{\cH_{i_0}^{+}(\delta)} = 1$. As any countable intersection of unit-probability events is also unit-probability, we have $\bP{\cH^{+}(\delta)} = 1$ where 
\begin{equation}
    \cH^{+}(\delta) \triangleq \cap_{i_0}\cH_{i_0}^{+}(\delta).
\end{equation}
Consider an $\omega \in \cG(\epsilon)\cap \cH^{+}(\delta)$, with $\cG(\epsilon)$ defined in \eqref{eq:set_G}. Repeated application of  \eqref{eq:matrix_form} yields
\begin{equation}
    \bnu_{i} = (1-\epsilon)^{\bm{i}_0-i}\bY_{\bm{i}_0}^i \bnu_{\bm{i}_0}, \quad \forall i \geq \bm{i}_0(\omega).
    \end{equation}
Furthermore, since $\omega \in \cH^{+}(\delta)$ implies 
 \begin{equation}
      [\bY_{\bm{i}_0}^i]_{\ell k} \leq e^{(i-\bm{i}_0)(\gamma+\delta)} 
 \end{equation}
for all $i \geq \bm{i}_1(\omega)$,  
 \begin{equation}
    \bnu_{k,i} \leq e^{(i-\bm{i}_0)(\gamma+\delta+\epsilon')}, \quad \forall i \geq \bm{i}_1(\omega)
\end{equation}
with some $\epsilon' \triangleq -\log(1-\epsilon)$. Thus, 
\begin{equation}
\limsup_{i \to \infty} \frac 1 i \log\bmu_{k,i} \leq \limsup_{i \to \infty} \frac 1 i \log\bnu_{k,i} \leq \gamma + \delta + \epsilon'
\end{equation}
for all $\omega \in \cG(\epsilon)\cap \cH^{+}(\delta)$. Since $\cG(\epsilon)$ and $\cH^{+}(\delta)$ are both probability one events, so is their intersection. This completes the part $(i)$ of the proof.

The part ($ii$) of the proof requires the construction of another extremal process $\{\bzet_{k,i}\}$, which lower bounds $\{\bmu_{k,i}\}$. Similar to the part $(i)$, we define the events 
\begin{equation}
    \cH_{i_0}^{-}(\delta) \triangleq \Big\{\omega \in \Omega: \exists i_1\geq i_0, \forall i \geq i_1, \frac 1 {i-i_0} \min_{\ell,k}\log [\bY_{i_0}^i]_{\ell k} \geq \gamma-\delta\Big\}
\end{equation}
and 
\begin{equation}
    \cH^{-}(\delta) \triangleq \cap_{i_0}\cH_{i_0}^{-}(\delta).
\end{equation}
Then, for any $\omega \in \cG(\epsilon)\cap \cH^{-}(\delta)$, we have $\bmu_{k,i}(\theta) \leq \epsilon$ for $i \geq \bm{i}_0(\omega)$ and 
\begin{align}
    \bmu_{k,i}(\theta) &\geq \sum_{\ell \in \cN_k} a_{\ell k} \dfrac{\bmu_{\ell,i-1}(\theta)\br_{\ell,i}(\theta)}{1+\epsilon \sum\limits_{\theta' \neq \theta^\circ} \br_{\ell,i}(\theta)} \notag \\
   & \geq \sum_{\ell \in \cN_k} a_{\ell k} \dfrac{\bmu_{\ell,i-1}(\theta)\br_{\ell,i}(\theta)}{1+\epsilon \sum\limits_{\theta' \neq \theta^\circ }\sum_{\ell} \br_{\ell,i}(\theta')}.
\end{align}
We introduce the vector $\bzet_i=[\bzet_{1,i}, \bzet_{2,i},\dots,\bzet_{K,i}]^{\T}$ and define its evolution for $i\geq \bm{i}_0$ as
\begin{align}
    \bzet_i = (A^{\T}\bR_i) \bzet_{i-1}, \quad\bzet_{\bm{i}_0} = \bmu_{\bm{i}_0}.
\end{align}
This implies that
\begin{align}
    \frac 1 i \log\bmu_{k,i} &\geq  \frac 1 i \log\bzet_{k,i}- \frac 1 i \sum_{j = i_0}^i \log\bigg(1+\epsilon \sum_{\theta' \neq \theta^\circ }\sum_{\ell} \br_{\ell,j}(\theta')\bigg) \notag \\
    &\stackrel{(a)}{\geq} \frac 1 i \log\bzet_{k,i}- \epsilon \frac 1 i \sum_{j = i_0}^i  \sum_{\theta' \neq \theta^\circ }\sum_{\ell} \br_{\ell,j}(\theta')
\end{align}
where \( (a) \) follows from $\log(1+x) \leq x$. 
Observe that by the strong law of large numbers
\begin{align}
    \lim_{i \to \infty} \epsilon \frac 1 i \sum_{j = i_0}^i  \sum_{\theta' \neq \theta^\circ }\sum_{\ell} \br_{\ell,j}(\theta') &= \epsilon K (H-1)\E[\br_{\ell,j}] \notag \\&= \epsilon K (H-1).
\end{align}
Hence, proceeding similarly to the previous part we obtain that almost surely
\begin{align}
    \liminf_{i \to \infty} \frac 1 i \log\bmu_{k,i} &\geq \liminf_{i \to \infty} \frac 1 i \log\bzet_{k,i} - \epsilon K (H-1) \notag \\ &\geq \gamma - \delta - \epsilon K (H-1).
\end{align}
Since $\delta$ and $\epsilon$ are arbitrary, the proof is complete.

\subsection{Conditions for the Subadditive Ergodic Theorem}\label{app:subadditive}
We first restate the subadditive ergodic theorem.
\begin{theorem}[\kern-1ex\cite{Kingman}, Theorem 1]\label{thm:subadditive}
Let $\{\bx_{ij}\}_{j\leq i}$ be a doubly indexed random sequence that satisfies the following conditions.
\begin{itemize}
    \item[($i$)] The distribution of $\bx_{ij}$ depends only on $j-i$.
    \item[($ii$)] $\bx_{ij} \leq \bx_{ik} + \bx_{kj}$, for all $i \leq k \leq j$.
    \item[($iii$)] $\frac 1 j \E[\bx_{1j}] \geq -\kappa$ for some constant $\kappa$ and for all $j\geq 1$.
\end{itemize}
Then, the finite limit
\begin{equation}
    \bm{\gamma} = \lim _{i \rightarrow \infty} \frac 1 i \bx_{1i}
\end{equation}
exists almost surely and in the mean, and furthermore,
\begin{equation}
    \E[\bm{\gamma}] = \lim _{i \rightarrow \infty} \frac 1 i \E[\bx_{1i}].
\end{equation}
\end{theorem}\putqed \vspace{-1em} \\

To use this result, we replace $\bx_{ij}$ with $-\log\|\bY_i^j\|_-$ (defined in \eqref{eq:matrix_product}). Since $\{A^{\T}\bR_i\}$ is an i.i.d. sequence, $\bY_i^j$ is stationary and ($i$) is satisfied. In the relation
\begin{equation}
    \|\bY_i^k\|_-\|\bY_k^j\|_- \leq \|\bY_i^j\|_{-},
\end{equation}
taking the logarithm and negating both sides shows immediately that it is ($ii$) is satisfied. The only remaining condition to verify is ($iii$), i.e., whether
\begin{equation}
    \E[-\log\|\bY_i\|_-]\geq -\kappa i
\end{equation}
for some constant $\kappa$. We achieve this by following similar steps to the proof of \cite[Theorem 5]{Kingman}:
\begin{align}
    \E[\log\|\bY_i\|_-]&\leq \E[\log\|\bY_i\|_1] \notag \\
    &\stackrel{(a)}{\leq} i \E[\log\|\bY_1\|_1] \notag \\
    &\leq i\E[\log(K \max_{k, \ell} [\bY_1]_{k \ell})] \notag \\
    &\leq i\bigg(\E\bigg[\sum_{k, \ell}(\log[\bY_1]_{k \ell})^+\bigg] + \log K\bigg) \notag\\
    &\stackrel{(b)}{\leq} i \kappa
\end{align}
where $(a)$ follows from subadditivity, and $(b)$ follows from Assumption 1. Hence, $-\log\|\bY_i\|_-$ satisfies the conditions of the subadditive ergodic theorem.
\subsection{Proof of Lemma \ref{lem:kkt}}\label{app:kkt}
For $v \in \mathbb{S}_K$, it is known that \cite{Gallager} for some finite constant $c$, the KKT conditions are given by 
\begin{align}\label{eq:KKT}
\fracpartial{F}{v_k} &= c,\quad v_k > 0\\
\fracpartial{F}{v_k} &\leq c,\quad v_k = 0.
\end{align}
We however have to justify the interchange of differentiation and expectation. Observe that
\begin{align}
\fracpartial{}{v_k}\log(\sum_\ell v_\ell \br_\ell) &=\lim_{\epsilon\to 0} \frac 1 \epsilon\bigg(\!\!\log(\sum_{\ell} \! v_\ell \br_\ell + \epsilon \br_k) \!- \!\log(\sum_\ell \! v_\ell \br_\ell)\!\!\bigg)\notag\\
&\leq \frac{\br_k}{\sum_{\ell} v_\ell \br_\ell} \label{eq:dominated} 
\end{align}
where we used the inequality $\log(1+x)\leq x$. Since all the random variables in \eqref{eq:dominated} are uniformly bounded by $\dfrac{\br_k}{\sum_{\ell} v_\ell \br_\ell}$, and $\E\big[\dfrac{\br_k}{\sum_{\ell} v_\ell \br_\ell}\big] < \infty$, the dominated convergence theorem \cite{williams_1991} justifies interchanging differentiation and expectation. Therefore,
\begin{equation}
    \fracpartial{F}{v_k} = -\frac{\pi_k}{v_k} + \E\bigg[\frac{\br_k}{\sum_{\ell} v_\ell \br_\ell}\bigg].
\end{equation}
Also, $\pi$, being the Perron vector of a primitive matrix $A$, has all strictly positive entries. Hence, setting any $v_k = 0$ will make $F(v)$ infinite and all $v_k$ of our interest must also have all positive entries. As a result, from \eqref{eq:KKT},
\begin{equation}
    -\frac{\pi_k}{v_k} + \E\bigg[\frac{\br_k}{\sum_{\ell} v_\ell \br_\ell}\bigg] = \mu.
\end{equation}
In this equation, multiplying both sides with $v_k$ and summing over each agent $k$, we get $\mu = 0$, which concludes the proof.
\subsection{Proof of Theorem \ref{thm:gap}}\label{app:jensen}
We first show that $\rho^{(\GA)} = \rho^{(\AA)}$ implies $\br_k=\br_\ell$ for all $k,\ell$. 
Applying Jensen's inequality to exchange the logarithm and expectation, we observe that 
\begin{equation}
    G(v) \geq \sum_k \pi_k\E\bigg[\log\frac{\br_k}{\sum_{\ell} v_\ell \br_\ell}\bigg] + \E\bigg[\log(\sum_k v_k \br_k)\bigg] + \rho^{(\GA)}  = 0.
\end{equation}
Since the logarithm is a strictly convex function, $G(v) = 0$ only when
\begin{equation}
    \log\E\bigg[\frac{\br_k}{\sum_{\ell} v_\ell \br_\ell}\bigg] = \E\bigg[\log\frac{\br_k}{\sum_{\ell} v_\ell \br_\ell}\bigg].
\end{equation} 
In other words, $G(v) = 0$ only when $\dfrac{\br_k}{\sum_{\ell} v_\ell \br_\ell}$ is constant with probability one for all $k$. This means that for any $k,\ell$,
\begin{equation}
    \frac{\br_k}{\sum_{\ell} v_\ell \br_\ell}\frac{\sum_{\ell} v_\ell \br_\ell}{\br_m} = \frac{\br_k}{\br_m}
\end{equation}
is also constant. Also, since by definition
\begin{equation}
    \E[\br_k]=\E[\br_m]=1, 
\end{equation}
$\br_k$ must be equal to $\br_m$ with probability one. This shows that if $\br_k \neq \br_\ell$ for some $k$, $\ell$ with non-zero probability, then $G(v)$ is strictly positive. Also, due to the fact that
\begin{equation}
    \inf_{u\in \mathbb{S}_K}\E[D(\pi||\bu_1)] \geq \inf_v F(v) \geq \inf_v G(v) > 0
\end{equation}
and $\eta_A < 1$, Eq. \eqref{eq:infimization} implies $\rho^{(\GA)} > \rho^{(\AA)}$.

The other direction is more straightforward to establish. If for all $k$, $\ell$, we have the relation $\br_k=\br_\ell = \br$, it holds that
\begin{align}
    [\bY_i]_{k\ell} &= \prod_{j = 1}^i \br_j [A^i]_{\ell k}.
\end{align}
This, in turn, implies
\begin{equation}
    \frac 1 i \log[\bY_i]_{k\ell} \to \E[\log \br] = -\rho^{(\GA)} 
\end{equation}
since every column of $A^i$ tends to its Perron vector $\pi$.
\bibliographystyle{IEEEtran}
\bibliography{ref.bib}

\end{document}